\documentclass[lettersize,journal]{IEEEtran}
\usepackage{amsmath,amsfonts}
\usepackage{algorithmic}
\usepackage{algorithm}
\usepackage{array}
\usepackage[caption=false,font=normalsize,labelfont=sf,textfont=sf]{subfig}
\usepackage{textcomp}
\usepackage{stfloats}
\usepackage{url}
\usepackage{verbatim}
\usepackage{graphicx}
\usepackage{cite}

\usepackage{xcolor}
\definecolor{myred}{rgb}{0.8,0.1,0.6}

\usepackage{ulem}

\usepackage{tcolorbox}
\usepackage{adjustbox}
\usepackage{booktabs}
\usepackage{amsthm}
\usepackage{multirow}
\usepackage{amssymb}
\usepackage{newtxmath}
\newtheorem{definition}{Definition}[section]

\newtheorem{lemma}{Lemma}[section]

\usepackage{enumitem}

\begin{document}

\title{Is External Database Protection Static in Retrieval-Augmented Generation? Rethinking Privacy Preservation under Dynamic Queries}
%Prompt-aware Dynamic Hierarchical Differential Privacy Protection for RAG
%Is External Database Protection Static in Retrieval-Augmented Generation? Rethinking Privacy Preservation under Dynamic Queries
%Prompt-aware Hierarchical Differential Privacy Protection for Retrieval-augmented generation

\author{Gang Zhang$^1$, Mingyu Tian$^1$, Xukun Luan$^1$, Yuanchi Ma$^2$, Jinyan Liu$^1$\\
$^1$\textit{School of Computer Science and Technology, Beijing Institute of Technology , Beijing, China}\\
$^2$\textit{Department of Computer Science and Technology, Tsinghua University, Beijing, China}}

        % <-this % stops a space
%~\IEEEmembership{Student,~IEEE,}
%\thanks{This paper was produced by the IEEE Publication Technology Group. They are in Piscataway, NJ.}% <-this % stops a space
%\thanks{Manuscript received April 19, 2021; revised August 16, 2021.}

\markboth{Journal of \LaTeX\ Class Files,~Vol.~14, No.~8, August~2021}%
{Shell \MakeLowercase{\textit{et al.}}: A Sample Article Using IEEEtran.cls for IEEE Journals}


\maketitle

\begin{abstract}
Retrieval-augmented generation (RAG) enhances large language models via external document retrieval,  but retrieved contexts may leak sensitive information. Current privacy protection methods typically rely on a document-level static risk assumption, treating all retrieved documents as having the same privacy leakage risk. However, this assumption overlooks a fundamental characteristic of RAG: the privacy risk of a document is highly dependent on the user's query, making privacy leakage inherently query-driven and dynamic.
To address this challenge, we propose a Prompt-Aware Dynamic Hierarchical Differential Privacy framework (PA-HDP) for privacy-preserving RAG. PA-HDP first performs a prompt-aware risk hierarchy to dynamically assess privacy risks under different queries. It then applies adaptive sensitive entity replacement and exponential mechanism-based text selection to provide differentiated privacy protection while preserving semantic utility. By protecting only the content that is truly sensitive under a given query, PA-HDP minimizes unnecessary modifications to the retrieval corpus.
Extensive experiments on benchmark datasets demonstrate that PA-HDP significantly reduces privacy leakage while maintaining high retrieval quality, achieving a better privacy–utility trade-off than prior methods.
%Extensive experiments on medical dialogue and open-domain question answering benchmarks demonstrate that PA-HDP effectively mitigates both targeted and untargeted privacy attacks while maintaining high retrieval utility. Moreover, under rigorous differential privacy guarantees, our method significantly reduces semantic distortion and achieves a superior trade-off between privacy protection and model performance. These results highlight the importance of query-aware privacy protection and provide a practical framework for deploying RAG systems on sensitive data.
\end{abstract}

\begin{IEEEkeywords}
Retrieval-augmented generation; Differential privacy; Privacy protection; Large language models.
\end{IEEEkeywords}

\section{Introduction}

The emergence of Retrieval-Augmented Generation (RAG) technology has significantly mitigated the core challenges prevalent in large language models (LLMs) \cite{gao2023retrieval}, such as factual inaccuracies and content hallucinations\cite{huang2025survey}. By deeply integrating external knowledge bases with LLMs, RAG first retrieves the most relevant document segments from the knowledge repository based on the user's query prior to response generation. These retrieved results are then fed into the model as contextual information, thereby guiding the model to produce more accurate and reliable outputs. Currently, RAG technology has been widely adopted in high-stakes domains where information accuracy is paramount, including healthcare \cite{amugongo2025retrieval}, finance\cite{zhao2024optimizing}, and government services \cite{hindi2025enhancing}.

However, when the external knowledge base of a RAG system contains a large amount of sensitive information, such as patient medical records and personally identifiable information (PII), the retrieved text segments often contain highly sensitive content \cite{10zeng2024good}. This sensitive information may be leaked by large language models through various means, including direct verbatim reproduction, indirect inference, and gradual reconstruction, leading to uncontrollable privacy risks. Multiple studies \cite{8zeng2025mitigating, 25guan2025privacy, 26mu2026towards} have pointed out that RAG systems face severe privacy and security challenges when the retrieval process involves private data. Therefore, enhancing the privacy protection capabilities of RAG systems and preventing the leakage of retrieved data are crucial for preventing unauthorized access and misuse of data \cite{4grislain2025rag}.

Differential Privacy (DP) \cite{dwork2025differential} offers a mathematically rigorous solution to this challenge. By injecting carefully calibrated noise into the data, DP guarantees that an adversary cannot determine with high confidence whether a specific individual's record is present in the dataset, even if they have access to all other records. 
Several existing works have leveraged DP to protect the privacy of retrieved contexts in RAG systems. Koga et al. \cite{1koga2025privacypreserving} proposed selective privacy budget allocation to sensitive tokens to support accurate long-form generation under constrained budgets. Grislain \cite{4grislain2025rag} further demonstrated the feasibility of applying DP to token-level generation in RAG. Beyond token-level protection, some works have also explored document-level privacy mechanisms\cite{5hemmat2025vague}. Mori et al. \cite{9mori2025differentially} proposed a differentially private synthetic text generation approach that replaces sensitive documents with DP-compliant synthetic alternatives, mitigating leakage at the data source.

The aforementioned studies have provided valuable insights into privacy protection for retrieved contexts in RAG systems. 
However, these methods share a fundamental limitation: they are all based on the document-level static privacy risk assumption, which presumes that all retrieved documents carry an equal risk of privacy leakage. This overlooks the critical fact that the privacy leakage risk of the same document is highly dependent on the specific content of the user's query (\textit{prompt}). For instance, for the same internal company document, a query asking to "summarize the company's cultural philosophy" poses almost no privacy risk, while a query requesting to "list the names and contact information of all department employees" will lead to severe privacy leakage. \textbf{This implies that the privacy protection process in RAG systems should essentially be a query-driven dynamic process rather than the document-level static process currently widely adopted.}

To address the inherent limitations of the aforementioned static privacy protection methods, this paper proposes a Prompt-Aware Dynamic Hierarchical Differential Privacy (PA-HDP) protection method. In particular, this method introduces a prompt-driven fine-grained privacy risk assessment mechanism into the retrieval phase of RAG systems. 
First, it performs sentence-level sensitive semantic recognition on the retrieved contexts based on the user's query (\textit{prompt}), and divides the text segments into different privacy risk levels according to the type of sensitive information, the severity of leakage harm, and the relevance to the query intent. 
Then, for text content at different risk levels, we design an adaptive candidate selection strategy based on the exponential mechanism, where different privacy budgets are assigned according to the risk level. Specifically, a diverse candidate set is first generated through semantically equivalent replacement of sensitive entities. The exponential mechanism then selects a protected text from the candidate set that satisfies differential privacy constraints while preserving the highest possible semantic fidelity.
This method can achieve an optimal trade-off between privacy security and model generation utility without compromising the semantic coherence and information integrity of the original context. 
Our experimental results show that using our algorithm can achieve comparable performance with using original data while substantially reducing the associated privacy risks.

Our contributions are summarized as follows :
\begin{itemize} 
    \item[1.] \textbf{An important yet underexplored research question has been raised.} We point out the prevalent flaw of document-level static risk assumption in RAG privacy protection and reveal the query-dependent dynamic privacy leakage problem.
    \item[2.] \textbf{A Prompt-Aware Dynamic Hierarchical Differential Privacy framework (PA-HDP) is proposed.}
    We propose a Prompt-Aware Dynamic Hierarchical Differential Privacy protection framework. The framework integrates risk stratification, adaptive sensitive entity replacement, and exponential mechanism-based text selection to provide differentiated privacy protection while preserving semantic utility.
    \item[3.] \textbf{The experiments verified the effectiveness of the proposed method.}
    Extensive experimental results demonstrate that, under the premise of strictly satisfying the mathematical constraints of differential privacy, the PA-HDP method can significantly reduce the semantic loss of generated content and achieves a better trade-off between privacy and model utility.

\end{itemize}

\section{Related work}

Retrieval-Augmented Generation (RAG) improves the factuality of large language models but introduces new privacy risks, as sensitive information contained in retrieved documents may be exposed through generated responses. Ensuring strong privacy guarantees is therefore critical, with differential privacy emerging as a principled and widely adopted framework for mitigating information leakage. Existing privacy-preserving RAG approaches can be broadly categorized into three directions: (1) token-level generation protection based on differential privacy, (2) protection based on external databases, and (3) other privacy protection methods.

\subsection{Privacy protection for Token-level Generation}

These methods directly apply differential privacy mechanisms during the LLM decoding phase, injecting noise into token probability distributions or logits to limit the leakage of retrieved information. Early work \cite{1koga2025privacypreserving}  explored the fundamental challenges of differentially private RAG, namely how to generate accurate long text under limited privacy budgets, and proposed a strategy that allocates budgets only to sensitive tokens. InvisibleInk \cite{2vinod2026invisibleink} interprets next-token sampling as an exponential mechanism over logits, reducing privacy costs by isolating sensitive information in public versus private logits, and samples from a superset of top‑k private tokens to improve utility, achieving an 8‑fold reduction in computational cost under the same privacy level. Privacy‑Aware Decoding (PAD) \cite{3wang2025privacy} is a lightweight inference-time defense that adaptively adds calibrated Gaussian noise to token logits, combining confidence filtering with efficient sensitivity estimation, protecting only high-risk tokens in a model-agnostic manner with very low computational overhead. A complementary work \cite{4grislain2025rag} further validates the feasibility of differential privacy based token generation in private RAG scenarios. PEARL \cite{6joopearl} proposes an adaptive differential privacy decoding framework based on confidence gaps and entropy regulation, dynamically allocating privacy budgets at both token and sentence levels, focusing protection on personally identifiable information fragments.

\subsection{Privacy protection for External Databases}

The core objective of these methods is to directly protect the retrieved original documents so that LLMs or attackers cannot access sensitive content. VAGUE‑Gate \cite{5hemmat2025vague} proposes a plug‑and‑play LDP protection mechanism that performs budget-constrained token filtering and random rewriting perturbation on retrieved text, hiding the content of original documents while satisfying $\varepsilon$‑LDP. LPRAG \cite{7he2025mitigating} applies LDP perturbation only to private entities (words, numbers, or phrases) within the text, identifying entity types and allocating adaptive privacy budgets, using different perturbation mechanisms for each type. Differentially private synthetic text \cite{9mori2025differentially} avoids leakage at the source by generating DP-compliant synthetic documents to replace sensitive original data. DP‑KSA \cite{12tang2026differentially} adopts keyword semantic compression and a propose‑test‑release mechanism to extract high‑frequency keywords from retrieved documents in a differentially private manner, augmenting only the compressed keywords into the prompt. Text‑DP \cite{13yu2024textual} proposes a textual differential privacy paradigm for context‑aware inference, achieving sensitive information anonymization through differential embedding hashing. RemoteRAG \cite{14cheng2025remoterag} formalizes the privacy protection problem for cloud‑based RAG services, introducing an $(n,\varepsilon)$-DistanceDP mechanism to perturb queries and document embeddings, narrowing the retrieval scope based on perturbed embedding vectors. ppRAG \cite{24ye2025efficient} combines distance‑preserving symmetric encryption (CAPRISE) with the DistanceDP mechanism to enable similarity computation over encrypted embeddings in untrusted cloud environments, ensuring that plaintext documents are never exposed.

Another class of methods does not rely on differential privacy, protecting original documents through synthetic data, knowledge distillation, embedding shifting, or model unlearning. Pure synthetic data methods \cite{8zeng2025mitigating} replace the retrieval database entirely with artificially generated text that contains no sensitive information. Parametric RAG \cite{11chen2025privacy} adopts knowledge distillation to convert each document into a parameterized LoRA module, requiring no access to original documents during inference, with document content masked by special tokens. PRESS \cite{15he2025press} fine‑tunes the model via embedding space shifting, enabling RAG inference without needing to access original sensitive documents. Additionally, \cite{19wang2025learning} studies methods to erase private knowledge from multiple documents, fine‑tuning the model to actively forget sensitive information from retrieved documents.

\subsection{Privacy protection for Training stage}

Beyond the direct protections targeting the generation process and retrieved documents, researchers have also worked on training‑stage privacy, multi‑query privacy accounting, and system security evaluation. Regarding user‑level privacy protection during the training phase of RAG systems, Charles et al.\cite{16charles2025learning} designs two scalable DP‑SGD variants (ELS example‑level sampling and ULS user‑level sampling) and derives a tight user‑level privacy accounting framework, adaptively selecting user group sizes between privacy and computational resources via data‑driven heuristics; Boenisch et al. \cite{17boenisch2023have} further proposes a personalized differential privacy method that assigns different sampling rates or clipping norms to different users to achieve individualized protection strength. In multi‑turn interaction scenarios, continuous privacy budget consumption is a key challenge for practical deployment. Private‑RAG (MURAG) \cite{18wu2025private} shifts privacy accounting from the query level to the document level through document access frequency‑based individual privacy filtering and adaptive threshold release, effectively controlling privacy accumulation in multi‑turn interactions. Beyond Per‑Query Privacy \cite{20wu2025beyond} systematically studies privacy budget consumption in online interactions, designing both fixed‑threshold and adaptive‑threshold private retrieval algorithms, making privacy consumption related to document retrieval frequency rather than total query count, achieving more efficient document selection and multi‑turn QA generation under $\varepsilon$‑DP guarantees.

%Furthermore, Comparative Analysis of RAG Frameworks \cite{21sah2026comparative} provides a systematic review and trade‑off analysis of frameworks such as LangChain and LlamaIndex, as well as mechanisms including differential privacy, secure multi‑party computation, and homomorphic encryption in private QA scenarios. \cite{10zeng2024good} empirically demonstrates the vulnerability of RAG systems to leaking private retrieval databases from an attack perspective. \cite{22arzanipour2025rag} constructs a formal threat model for RAG, defining specific attack vectors such as document‑level membership inference and content leakage, providing a unified theoretical benchmark and analytical framework for the various defense methods discussed above.

\section{Preliminary}
%For fixed privacy budget $\epsilon$ and $\delta$, 

Differential privacy is a standard paradigm for protecting privacy of individuals. It requires that changing one entry can create only a small change of the output distributions. %records in a statistical analyse. It formalizes the requirement that the addition or removal of a record does not change the probability of the output of a randomized mechanisms by much.

\begin{definition}
    (Differential Privacy (DP)~\cite{dwork2006differential}). A randomized mechanism $\mathcal{M}$: $\mathcal{X}^{n} \to \mathcal{R}$ satisfies $\left ( \epsilon ,\delta  \right )$-differential privacy  if for any two adjacent inputs $\mathcal{D}$, $\mathcal{D}'$ %(differing in at most the data of one individual) 
    and for any subset of outputs $S\subseteq \mathcal{R}$, it holds that 
$$
\Pr[\mathcal{M}\left ( \mathcal{D} \right )\in S] \leq e^{\epsilon }P_{r}[\mathcal{M}\left ( \mathcal{D}' \right )\in S]+\delta ~.
$$
%where $\epsilon$ is the parameter of privacy protection budget, which is one of the key factors to measure the level of privacy protection. The smaller the parameter $\epsilon$ represents the higher the level of privacy protection.
%If $\delta=0$, it is also said as pure differential privacy.
\end{definition}

\begin{lemma}[Laplace Mechanism \cite{dwork2006calibrating}]
Given a function $f:\mathcal{X}^{n}\rightarrow\mathbb{R}$ with sensitivity $\Delta f = \max_{\mathcal{D},\mathcal{D}' \text{ adjacent}} \big|f(\mathcal{D})-f(\mathcal{D}')\big|$, the Laplace mechanism outputs
\[
\mathcal{M}(\mathcal{D}) = f(\mathcal{D}) + \eta,
\]
where $\eta \sim \mathrm{Lap}\left(\frac{\Delta f}{\epsilon}\right)$.
Then $\mathcal{M}$ satisfies $\epsilon$-differential privacy.
\end{lemma}

\begin{lemma} (Exponential Mechanism \cite{mcsherry2007mechanism}) 
For a dataset $\mathcal{D}$, let $u(\mathcal{D},r)$ be a utility function with respect to output $r\in\mathcal{R}$, and let $\triangle u$ be the sensitivity of $u$.
The exponential mechanism $\mathrm{M}(\mathcal{D},u,\mathcal{R})$ selects an element $r\in\mathcal{R}$ with probability proportional to
\[
\exp\left( \frac{\epsilon \, u(\mathcal{D},r)}{2\triangle u} \right),
\]
and this mechanism satisfies $\epsilon$-differential privacy.

\end{lemma}

To adapt differential privacy to textual data, we introduce Semantic Metric Differential Privacy (Definition~\ref{de:sdp}). 

\begin{definition}\label{de:sdp}[$(\epsilon,\delta)$-Semantic Metric Differential Privacy]
Let $\mathcal{X}$ denote the text space, where each text
$x\in\mathcal{X}$ is represented as a token sequence $x=(w_1,\dots,w_n)$.
Define an edit-based adjacency relation $x\sim x'$ such that
$x'$ can be obtained from $x$ by a single token- or phrase-level
insertion, deletion, or substitution operation, i.e., $d_{\mathrm{edit}}(x,x')=1.$
Let $\phi:\mathcal{X}\rightarrow \mathbb{R}^d$ be a semantic embedding function, and define the semantic metric
\[
d_{\phi}(x,x')
=
\|\phi(x)-\phi(x')\|_2.
\]

A randomized mechanism $M:\mathcal{X}\rightarrow\mathcal{Y}$ is said to satisfy $(\epsilon,\delta)$-Semantic Metric Differential Privacy if the following conditions hold:
For any adjacent texts $x\sim x'$ and any measurable set
    $S\subseteq\mathcal{Y}$,
    \[
    \Pr[M(x)\in S]
    \le
    e^{\epsilon d_{\phi}(x,x')}
    \Pr[M(x')\in S]
    +\delta.
    \]

\end{definition}

Based on this definition, we further propose the Semantic Laplace Mechanism and the Semantic Exponential Mechanism, which serve as the core privacy-preserving components of our method.

\section{Method}

This paper proposes a Prompt-Aware Dynamic Hierarchical Differential Privacy (PA-HDP) protection method. This method decomposes the entire privacy protection pipeline into three tightly coupled stages: the normalization and definition of privacy units, prompt-aware risk assessment, hierarchical DP protection.

\begin{algorithm}[htbp]
\caption{PA-HDP}
\label{alg:pahdp}

\textbf{Input:} Prompt $q$, retrieval document $D= \left \{d_1, \dots, d_n \right \}$, retrieval number $K$, size of candidate set $m$,  privacy budget $\epsilon$, and parameters $\lambda, \alpha \in [0,1]$.

\begin{algorithmic}[1]
\STATE //Retrieval and text preprocessing//
\STATE Encode document via embedding model to obtain dense vectors $\left \{E(d_1), \dots, E(d_n) \right \}$
\STATE Build vector index database for efficient similarity search
\STATE Encode user query to obtain $E(q)$
\FOR{each document chunk \(d_i\) in vector index}
    \STATE Compute cosine similarity for ($E(d_i), E(q)$)
\ENDFOR
\STATE Select top-\(K\) most relevant contexts $\left \{s_1, \dots, s_K \right \}$
\STATE Divide text segment $s_{i,j}$ for $s_i$.

\STATE //prompt-aware risk assessment//
\FOR{each text segment $s_{i,j}$}
    \STATE Encode text segment to obtain $E(s_{i,j})$
    \STATE Compute risk-aware semantic similarity $S_{rel}$
    \STATE Obtain sensitive field set \(E_{i,j} = \{e_{i,j}^1, \dots, e_{i,j}^m\}\) 
    \STATE Compute field sensitivity $S_{sen}$ 
    \STATE Compute risk score $S_{total}(q, s_{i,j})$ by Eq.\ref{eq:risk-score}
    \STATE Obtain noisy risk score 
    $$S^\ast_{total}(q, s_{i,j}) = S_{total}(q, s_{i,j}) + Lap(\frac{2}{\epsilon})$$
    \STATE Divide privacy risk levels according to $S^\ast_{total}(q, s_{i,j})$
\ENDFOR

\STATE //Hierarchical DP protection//
\FOR{each text segment $s_{i,j}$}
    \IF{$S^\ast_{total}(q, s_{i,j}) \ge \tau_1$}
        \STATE Generate candidate texts $C=\left \{s_{i,j}^{1}, \dots, s_{i,j}^{m} \right \}$ by replacing entities
        \STATE Use the exponential mechanism to privately output the optimal $s_{i,j}^\ast \in C$ with $\epsilon_H/\epsilon_M$
    \ENDIF
\ENDFOR
\STATE Replace $s_{i,j}$ with $s^\ast_{i,j}$ to get $\left \{s^\ast_1, \dots, s^\ast_K \right \}$

\RETURN Private contexts $\left \{s^\ast_1, \dots, s^\ast_K \right \}$

\end{algorithmic}
\end{algorithm}

Specifically, the first stage performs standardized data preprocessing and defines the sentence as the minimum granularity unit for privacy protection.
The second stage performs fine-grained privacy risk assessment and hierarchical classification on these retrieved segments by analyzing their semantic correlation with user queries. 
In the final stage, the framework first generates semantically equivalent candidate texts by applying differentiated sensitive entity replacement strategies. It then performs hierarchical differential privacy protection by selecting the optimal sanitized text from the candidate set through the exponential mechanism, while rigorously satisfying the formal constraints of differential privacy.
The complete pseudo-code is shown in Algorithm~\ref{alg:pahdp}.

Furthermore, Fig.~\ref{fig:fig1} provides an example illustrating the pipeline of our privacy protection algorithm.

\begin{figure*}[t]
    \centering
    \includegraphics[width=0.85\textwidth]{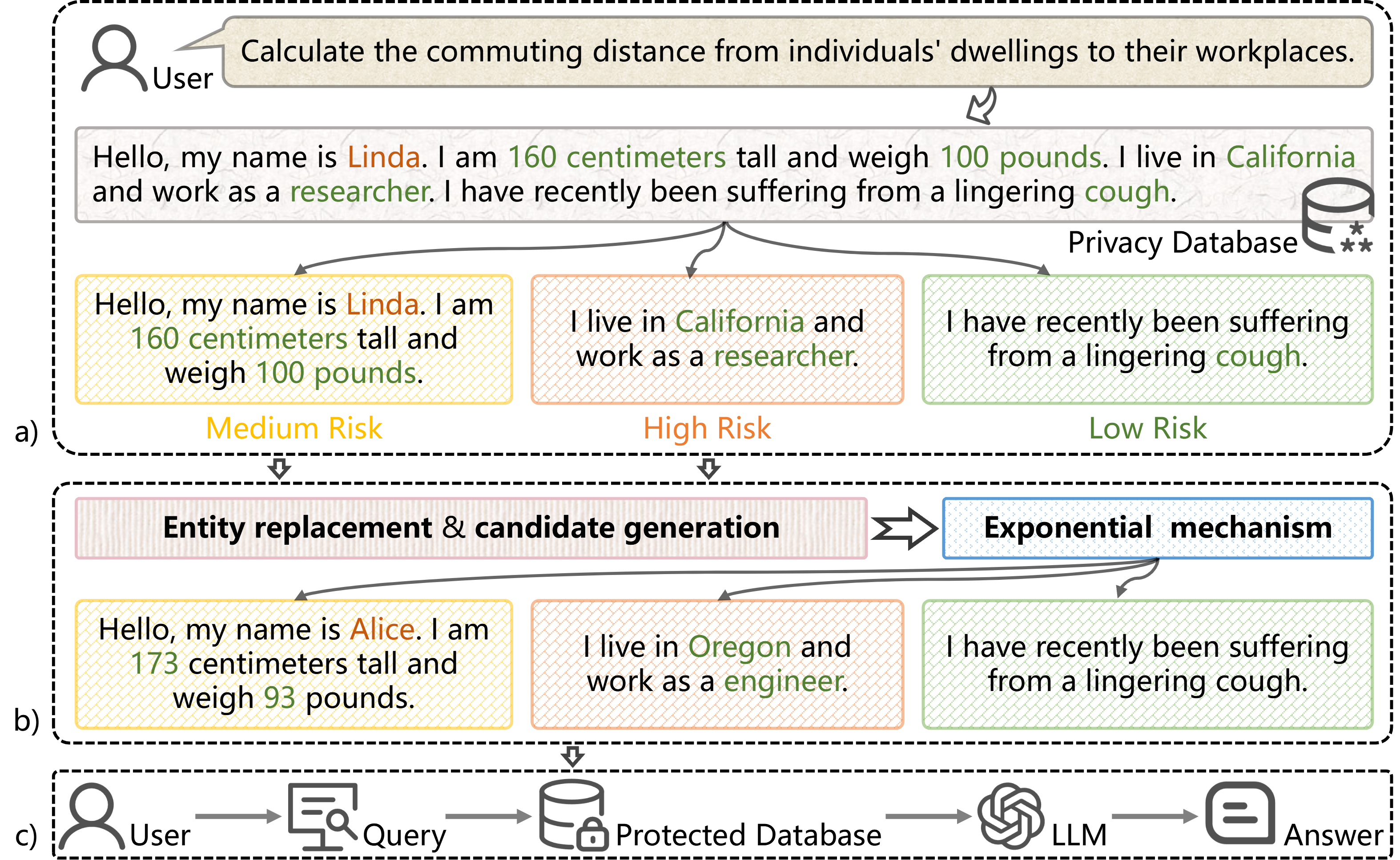}
    \caption{Pipeline of PA-HDP. a) Prompt-aware risk assessment stage. b) Hierarchical differential privacy protection stage. c) Full retrieval-augmented generation pipeline.}
    \label{fig:fig1}
\end{figure*}

%This paper employs a dense vector-based semantic retrieval approach. First, all documents $D= \left \{d_1, \dots, d_n \right \}$ in the external knowledge base are chunked, encoded $(E(\cdot))$, and used to construct a vector index $E(d_i)$. Upon receiving a user query q, the system converts it to a dense vector $E(q)$ using the same embedding model, computes cosine similarity \(\text{cos}(E(q), E(d_i))\) for semantic matching, and retrieves the top-K relevant contexts $\left \{s_1, \dots, s_K \right \}$. A fixed K threshold effectively controls context length, ensures necessary knowledge coverage, avoids semantic redundancy, and restricts privacy protection to the minimal query-relevant candidate set.

\subsection{Retrieval and Text Preprocessing}

In Retrieval-Augmented Generation (RAG), the retrieval stage determines the external knowledge accessible to the large language model (LLM) and serves as the primary entry point for privacy leakage. Given an external knowledge base $D=\{d_1,\dots,d_n\}$, each document is first chunked, embedded using an encoder $E(\cdot)$, and indexed as dense vectors. For a user query $q$, the system computes its embedding $E(q)$ and retrieves the top-$K$ text segments $\{s_1,\dots,s_K\}$ based on cosine similarity. The fixed top-$K$ strategy limits context length, reduces semantic redundancy, and confines privacy protection to query-relevant content.

To improve processing quality, the retrieved text is further preprocessed through sentence segmentation and text normalization, including punctuation-based boundary detection and the removal of redundant symbols, spaces, and line breaks. These operations provide standardized inputs for subsequent sensitive entity identification and differential privacy protection.

After preprocessing, each sentence $s_{i,j}$ is defined as a \textit{Basic Privacy Unit (PU)}, where $s_{i,j}$ denotes the $j$-th sentence in the $i$-th retrieved context. A privacy unit is the minimum processing unit for sensitive entity recognition, sensitivity assessment, candidate generation, and differential privacy perturbation. Each PU preserves semantic completeness while maintaining an independent privacy level and budget.

The independent design of privacy units enables fine-grained privacy protection and naturally satisfies differential privacy composition. It allows stronger protection for high-risk content while preserving the utility of low-risk text, thereby achieving a better privacy--utility trade-off.

\subsection{Prompt-aware risk assessment}
After completing retrieval and text preprocessing, to achieve fine-grained differentiated privacy protection, it is necessary to perform accurate localization, quantitative assessment, and hierarchical classification of sensitive information in the retrieved segments. In RAG scenarios, only sensitive information relevant to the current user query constitutes genuine privacy leakage risks. Even if query-irrelevant sensitive content exists in the retrieval results, it is difficult for Large Language Models (LLMs) to attend to and output it, thus requiring no high-intensity protection.

To realize prompt-aware dynamic risk assessment, this paper proposes a two-dimensional risk quantification model. It separately calculates the risk-aware semantic relevance \(S_{rel}\) between the retrieved segment and the user query $q$, and the field sensitivity \(S_{sen}\) of the segment itself. The final comprehensive risk score \(S_{total}\) is obtained through weighted fusion, based on which the text is divided into different privacy risk levels.

\subsubsection{Risk-aware semantic similarity}

We design a risk-aware gating mechanism that restricts semantic relevance computation to cases where the user query explicitly involves potentially sensitive entities. Specifically, we first assess whether the query contains potentially sensitive entities (e.g., person names, locations, or organizations). This assessment determines whether the current query is likely to trigger privacy leakage in the retrieved documents.

Let the user query be \(q\), and the j-th segment of the retrieved context $s_i$ be \(s_{i,j}\).
We first apply an entity recognition module \(\mathcal{E}(\cdot)\) to extract entities from the query:

\[
\mathcal{E}(q) = \{e_1, e_2, ..., e_m\}.
\]

We then define a sensitive-entity indicator function:

\[
I(q) =
\begin{cases}
1, & \text{if } \mathcal{E}(q) \cap \mathcal{E}_{sens} \neq \emptyset \\
0, & \text{otherwise}
\end{cases}
\]
where \(\mathcal{E}_{sens}\) denotes a predefined set of sensitive entity types.

When \(I(q) = 1\), we compute semantic similarity between the query and each segment using the same embedding model \(E(\cdot)\) employed in the retrieval stage:

\[
\mathbf{v}_q = E(q), \quad \mathbf{v}_{s_{i,j}} = E(s_{i,j}),
\]

\[
S_{rel}(q, s_{i,j}) = \frac{\mathbf{v}_q \cdot \mathbf{v}_{s_{i,j}}}{\|\mathbf{v}_q\| \cdot \|\mathbf{v}_{s_{i,j}}\|}.
\]

When \(I(q) = 0\), we conservatively assign $S_{rel}(q, s_{i,j}) = 0$, where \(\cdot\) denotes the vector dot product and \(\|\cdot\|\) denotes the L2 norm. Thus, \(S_{rel}(q, s_{i,j}) \in [0, 1]\).

Intuitively, when such entities are present, a higher semantic similarity between the query and a retrieved segment indicates a higher probability that the segment may be leveraged by the generation model, thereby increasing the risk of privacy leakage. Conversely, when no sensitive entities are detected in the query, we assume negligible entity-driven leakage risk from direct semantic matching and set the relevance score to zero.

%This formulation transforms traditional uniform similarity computation into a **conditional privacy-aware relevance model**, improving robustness against spurious similarity signals while focusing the risk estimation on entity-sensitive query scenarios.

\subsubsection{Field sensitivity}

Field sensitivity measures the inherent leakage harm degree of sensitive information contained in a retrieved text segment. This paper adopts a hybrid strategy combining Named Entity Recognition (NER) and rule-based extraction to comprehensively extract sensitive fields from retrieved text segments, forming a sensitive field set \(\mathcal{E}_{i,j} = \{e_{i,j}^1, \dots, e_{i,j}^m\}\), where \(e_{i,j}^k\) denotes the k-th sensitive entity in the $s_{i,j}$ text segment.
The NER module is implemented based on a pre-trained semantic model, which can automatically recognize and locate unstructured sensitive entities such as personal names, locations, ages, and disease diagnoses in the text. These entities usually have the ability to directly or indirectly identify user identities and are the primary sources of privacy leakage. On this basis, the rule-based extraction module complements the recognition of highly structured sensitive information such as phone numbers, email addresses, ID card numbers, and passwords through regular expressions, keyword libraries, and sensitive pattern matching. Such information cannot be fully recognized solely by semantic features but has fixed formats and combination patterns. This hybrid strategy ensures that all content that may pose privacy risks is included in the assessment scope, avoiding protection failure caused by omissions.
Subsequently, we assign predefined risk weights \(w_k \in [0, 1]\) to different types of sensitive entities according to their leakage harm degrees. For example, high risk weights (\(w_k \in [0.7, 1.0]\)) are assigned to entities such as ID card numbers, phone numbers, personal names, and disease diagnoses, while low risk weights (\(w_k \in [0.1, 0.3]\)) are assigned to entities such as locations, dates, and common nouns. The field sensitivity of the j-th segment is defined as the maximum value of the risk weights of all sensitive entities in the text segment:

$$S_{sen}(s_{i,j}) = \max_{1 \leq k \leq m} w_k.
$$

Using the maximum value instead of the average value as the field sensitivity ensures that a text segment will be classified as high risk as long as it contains one high-risk sensitive entity, thereby avoiding the underestimation of high-risk content caused by the averaging effect.

%打分情况，高分低分，附录

\subsubsection{Risk assessment}
The semantic relevance and field sensitivity are weighted and fused to obtain the comprehensive risk score of the j-th segment:
\begin{equation}\label{eq:risk-score}
    S_{total}(q, s_{i,j}) = \lambda \cdot S_{rel}(q, s_{i,j}) + (1-\lambda) \cdot S_{sen}(s_{i,j}).
\end{equation}
Where \(\lambda \in [0, 1]\) is the balance coefficient used to adjust the relative importance of semantic relevance and field sensitivity in risk assessment. In our experiments, the optimal value \(\lambda=0.6\) is determined through grid search, indicating that semantic relevance has a slightly greater impact on the final risk than the inherent sensitivity of the field itself.

Furthermore, we introduce the semantic Laplace mechanism to achieve rigorous privacy protection, and yield the noisy risk scores $S^\ast_{total}(q, s_{i,j})$. (The details are presented in the appendix.)

Finally, text segments are divided into three privacy risk levels according to their comprehensive risk scores:

$$\text{RiskLevel}(s_{i,j}) = 
\begin{cases} 
\text{Low Risk}, & S^\ast_{total}(q, s_{i,j}) < \tau_1 \\
\text{Medium Risk}, & \tau_1 \leq S^\ast_{total}(q, s_{i,j}) < \tau_2 \\
\text{High Risk}, & S^\ast_{total}(q, s_{i,j}) \geq \tau_2 
\end{cases}
$$

where $\tau_1$ and $\tau_2$ denote the risk-level partition thresholds. In our experiments, we set \(\tau_1=0.3\) and \(\tau_2=0.7\) based on empirical observations. Text segments (sentences) with different risk levels will adopt differential privacy protection strategies with different budgets to achieve the differentiated privacy protection objective of "low risk, low protection; high risk, high protection".

\subsection{Hierarchical DP protection}

To achieve provable and fine-grained privacy protection, we propose the semantic metric differential privacy as the refined privacy definition, and adopt the exponential mechanism as the core privacy algorithm.
This section sequentially performs candidate text generation, utility function design, semantic exponential mechanism construction and hierarchical DP protection, and finally outputs secure and highly usable sanitized text under the premise of strictly satisfying differential privacy constraints.

\subsubsection{Candidate text generation}

For each retrieved text segment $s_{i,j}$, we construct a set of semantically equivalent and structurally consistent candidate texts, providing a secure candidate space for the exponential mechanism to select the optimal output under differential privacy constraints. The candidate text construction adopts a core strategy centered on sensitive entity replacement, generating multiple semantically equivalent valid candidate texts while strictly preserving the overall semantic structure and non-sensitive content of the original sentence.

Specifically, we pre-construct a large-scale sanitized entity corpus $R$ as the secure material source for sensitive information replacement. We employ the hybrid entity extraction strategy proposed in Section 3.2.2 (combining Named Entity Recognition (NER) and rule-based extraction) to extract all entities from the external knowledge base and perform category annotation. 

For each retrieved text segment, the system first detects all sensitive entities, including their positions and categories. Next, it randomly selects entities of the same category from the sanitized entity corpus $R$ to perform isomorphic replacement. Ultimately, the system produces $m$ semantically identical candidate segments, denoted as $C=\left \{s_{i,j}^{1}, \dots, s_{i,j}^{m} \right \}$. We set $m=50$ in our experiments.

\subsubsection{Utility function design}

To privately select the most appropriate candidate text from a candidate set while preserving semantic fidelity and reducing privacy leakage risk, we design a composite utility function for the exponential mechanism. The utility function jointly considers two objectives: semantic preservation and privacy safety.
Specifically, given a retrieved segment $s_{i,j} \in \mathcal{X}$ and a candidate text $s_{i,j}^v \in C$, the utility function is defined as:
$$ u(s_{i,j}, s_{i,j}^v) = - \alpha \cdot d_{\text{sem}}(s_{i,j}, s_{i,j}^v) - (1-\alpha) \cdot r(s_{i,j}, s_{i,j}^v) 
$$

where $d_{\text{sem}}(s_{i,j}, s_{i,j}^v)$ denotes the semantic distance between the original segment and the candidate text, $r(s_{i,j}, s_{i,j}^v)$ represents the privacy leakage risk of the candidate text, $\alpha >0$ is a trade-off parameter balancing semantic quality and privacy protection, and $D>0$ is a normalization constant used to stabilize the utility scale and control sensitivity.

To measure semantic preservation, both the original segment and the candidate text are encoded into dense semantic embeddings using a pretrained text embedding model $\phi$. The semantic distance is then computed using the Euclidean distance in the embedding space:
$$d_{\text{sem}}(s_{i,j}, s_{i,j}^v) = \|\phi(s_{i,j})-\phi(s_{i,j}^v)\|_2.
$$
A smaller semantic distance indicates that the candidate text preserves the semantic meaning of the original text more effectively. Therefore, the utility function assigns higher utility values to candidates with smaller embedding distances.

To evaluate privacy security, we first extract all sensitive entities present in each candidate text and form a set, denoted as \(E(s_{i,j}^v)\). Each element in $E(s_{i,j}^v)$ denotes a potentially sensitive entity or phrase, such as names, addresses, phone numbers, or medical identifiers. If no sensitive span is detected, the privacy risk $r(s_{i,j}, s_{i,j}^v)=0$. Otherwise, each sensitive entity is embedded into the same semantic space, and its semantic similarity to the original segment is measured using cosine similarity. 
The overall privacy leakage risk of the candidate text is defined as the maximum similarity among all remaining sensitive spans:
$$r(s_{i,j}, s_{i,j}^v) = \max_{e\in E(s_{i,j}^v)} \text{cos} (s_{i,j}, e).
$$
This worst-case formulation emphasizes the most privacy-sensitive residual entity, where even a single highly correlated sensitive span may lead to successful information recovery. In addition, the parameter $\alpha$ controls the relative importance of privacy protection versus semantic preservation: larger values favor candidates with lower privacy leakage risk, while smaller values place greater emphasis on semantic fidelity. In our experiments, we set $\alpha=0.5$, which was found to provide the best trade-off between semantic fidelity and privacy protection.

It is worth noting that entity-level privacy risk assessment is also performed on the retrieved progress. However, the two stages serve different purposes: the former evaluates the privacy leakage risk introduced by the prompt itself, whereas the latter measures the residual privacy risk of candidate texts with respect to the original segment.
Additionally, to properly bound the semantic distance in the embedding space, all embedding vectors are normalized to unit norm, i.e., $||\phi(s)||_2 \le 1$. This normalization constrains the embedding space to a bounded domain, ensuring numerical stability and facilitating the derivation of the utility sensitivity bound under the semantic metric.

\subsubsection{Semantic exponential mechanism construction}
In this section, we first propose the Semantic Metric Differential Privacy, providing a theoretical basis for privacy protection methods.

Based on the above utility function, the exponential mechanism selects a candidate text according to the following probability distribution:

\begin{lemma}[Semantic Exponential Mechanism] \label{lem:exp}
Let $\mathcal C$ be a finite candidate text set, and let
\[
u(s_{i,j}, s_{i,j}^v)
=
- \alpha \cdot d_{\mathrm{sem}}(s_{i,j}, s_{i,j}^v)
- (1-\alpha) \cdot r(s_{i,j}, s_{i,j}^v),
\]
where $\alpha \ge 0$.
The Semantic Exponential Mechanism selects an output
$s_{i,j}^{\ast}\in\mathcal C$ according to the probability distribution
\[
\Pr[M(s_{i,j})=s_{i,j}^{\ast}]
=
\frac{
\exp\left(
\frac{\epsilon u(s_{i,j},s_{i,j}^{\ast})}
{2}%\Delta_u
\right)
}{
\sum\limits_{s_{i,j}^v\in\mathcal C}
\exp\left(
\frac{\epsilon u(s_{i,j},s_{i,j}^v)}
{2}
\right)
}.
\]

Then $M$ satisfies $(\epsilon,0)$-Semantic Metric Differential Privacy.
\end{lemma}
The proof details are presented in the appendix.

\subsubsection{Hierarchical DP protection}

After determining the candidate text set and privacy mechanism, the reasonable allocation of the privacy budget is the key to ensuring the balance between the protection effect of differential privacy and the utility of text generation.
To provide differentiated privacy protection for texts with varying sensitivity levels, we first classify all texts into three privacy risk categories: high-risk, medium-risk, and low-risk. High-risk texts typically contain highly sensitive information, while medium-risk texts contain partially sensitive or identifiable content. Low-risk texts are not considered in the subsequent privacy budget allocation.

After risk classification, differential privacy budgets are allocated separately to high-risk and medium-risk texts. Let the total privacy budget be $\epsilon$, the number of high-risk texts be $N_h$, and the number of medium-risk texts be $N_m$. Since high-risk texts require stronger privacy protection, they are assigned smaller privacy budgets than medium-risk texts.

We adopt a risk-aware budget allocation strategy defined as:

\[
\varepsilon_H
=
\frac{\varepsilon}{2(N_H + \lambda N_M)}, \ \varepsilon_M
=
\gamma \varepsilon_H
\]

where $\varepsilon_H$ and $\varepsilon_M$ denote the privacy budgets for high-risk and medium-risk texts, respectively, and $\gamma > 1$ controls the protection gap between different risk levels. In our experiments, we set $\gamma = 2$, meaning that medium-risk texts receive twice the privacy budget of high-risk texts. This strategy enables stronger protection for highly sensitive texts while preserving better semantic utility for medium-risk texts.

\section{Experiment}
This section verifies the effectiveness and privacy protection performance of the proposed PA-HDP method through a series of experiments. We first describe the experimental setup in Section A. Subsequently, Sections B and C evaluate the utility and privacy outcomes of the proposed method. Furthermore, we conduct an ablation study in Section D to investigate how different privacy budgets and numbers of retrieved documents affect the performance of our method. Finally, we perform a comprehensive hyperparameter comparison in Section F to identify the optimal hyperparameter configuration.

\subsection{Experimental details}
In this section, we introduce the evaluation datasets, baseline algorithms, and experimental setup.
The overall experimental evaluation followed the experimental setup in work \cite{8zeng2025mitigating}.

\subsubsection{Datasets}
We evaluate the effectiveness of our proposed methods under two privacy-relevant experimental scenarios. The first scenario targets real-world medical dialogue monitoring, where we adopt the HealthcareMagic-101 dataset (comprising 200,000 doctor-patient medical dialogues) as the retrieval corpus. 

In the second scenario, following Huang et al. \cite{ex-huang2023privacy}, we consider a setting where private information is mixed with public data.  Specifically, we construct the Wiki-PII dataset by embedding personally identifiable information (PII) fragments extracted from the private Enron Mail dataset into each sample of the public WikiText-103 dataset. We then benchmark our methods on four standard open-domain question answering (ODQA) benchmarks: Natural Questions (NQ)\cite{nq-kwiatkowski2019natural}, TriviaQA (TQA)\cite{tqa-joshi2017triviaqa}, Web Questions (WQ) \cite{wq-berant2013semantic}, and CuratedTrec (CT)\cite{ct-baudivs2015modeling}. 
Full details on dataset construction and characteristics are provided in Appendix A.9.

\subsubsection{Baselines}
To evaluate the effectiveness of our method, we compare against four baselines. The first is \textbf{Paraphrasing} \cite{8zeng2025mitigating}, which uses an LLM to extract the most relevant parts of the original context without performing any semantic rewriting. 
The second is \textbf{ZeroGen} \cite{bs-ye2022zerogen}, a zero-shot learning framework that first generates a synthetic dataset from scratch using a large pre-trained language model guided by task-specific prompts, then trains a tiny task model (e.g., LSTM) on the synthesized data for efficient inference.
The third is \textbf{AttrPrompt} \cite{bs-yu2023large}, a data synthesis approach that generates training data using diversely attributed prompts (specifying attributes such as length, style, or topic), which yields more diverse and less biased synthetic datasets compared to simple class-conditional prompts.
The fourth is \textbf{SAGE} \cite{8zeng2025mitigating}, a two-stage synthetic data generation paradigm that mitigates the privacy issues by pure synthetic data. In this method, \textbf{Stage-1} refers to the outputs generated from attribute-based generation, while \textbf{Stage-2} represents the outputs of the complete SAGE pipeline.
The fifth is \textbf{LPRAG } \cite{7he2025mitigating}, a privacy-preserving RAG
framework with formal privacy guarantees based on local differential privacy (LDP).
Detailed descriptions of these baselines are provided in Appendix A.3.

In addition, we report several variants of our framework to analyze the contribution of each component. \textbf{0-shot} denotes generation without retrieval augmentation. \textbf{Origin} uses the original dataset directly as the retrieval corpus. 

\subsubsection{Experimental Setup}

For both utility and privacy evaluations, we employ two representative generation backbones: the closed-source \textit{GPT-3.5-Turbo} and the open-source \textit{Llama3-8B-Chat (L8C)}. Using both proprietary and open-source models allows us to assess the generality of our approach across different model families and deployment settings. We include L8C because its performance on the target tasks relies heavily on retrieval augmentation, making it suitable for evaluating the utility gains provided by RAG. In addition, both models have undergone safety alignment, enabling us to investigate privacy risks in standard RAG pipelines and evaluate the effectiveness of our proposed protection mechanism under realistic deployment conditions.

We adopt \textit{bge-large-en-v1.5} as the embedding model. Unless otherwise specified, similarity is computed using the Cosine similarity between normalized embeddings, GPT-3.5-Turbo is adopted as the generation model, only the top retrieved document ($k=1$) and a privacy budget of $\epsilon=5$ are provided for each query.

\subsection{Utility evaluation}

To evaluate the utility of synthetic data as retrieval corpora, we assess the quality of generated answers by comparing them against the corresponding ground-truth answers. We report two widely used text generation metrics, \textbf{ROUGE-L} and \textbf{BLEU}, as the primary evaluation measures.

ROUGE-L evaluates the longest common subsequence between the generated answer and the reference answer, capturing their similarity in terms of content coverage and sentence-level structure. BLEU measures the overlap of n-grams between the generated and reference answers, reflecting the lexical accuracy and fluency of the generated text. Higher ROUGE-L and BLEU scores indicate better agreement with the ground truth and, consequently, higher utility of the synthetic retrieval data.

\subsubsection{Results on Medical Dialog}

\begin{table}[t]
\label{tab:res-health}
\caption{Utility results on HealthCareMagic dataset}
\centering
\begin{tabular}{c c c c c}
\toprule
Method & \multicolumn{2}{c}{GPT-3.5} & \multicolumn{2}{c}{Llama3-8b-Chat} \\
\cmidrule(lr){2-3} \cmidrule(lr){4-5}
& BLEU & ROUGE-L & BLEU & ROUGE-L \\
\midrule
0-shot & 0.1143 & 0.1045 &      0.081 & 0.0765\\
Origin & 0.1193 & 0.1078 &     0.0846 & 0.0789\\
Paraphrase & 0.1481 & 0.1303 &  0.105 & 0.0952\\
ZeroGen & 0.1199 & 0.1050 &    0.085 & 0.0769\\
LPRAG &  0.3186 & 0.3310 & \textbf{0.1634} & 0.1596 \\
AttrPrompt & 0.1114 & 	0.0915 &   0.079 & 0.0670\\
Stage-1 & 0.1578 & 0.1306 &     0.114 & 0.0956\\
Stage-2 & 0.1594 & 0.1288 &     0.1192& 0.1160\\
Ours & \textbf{0.3275} & \textbf{0.3445}&        0.1512& \textbf{0.1726}\\
\bottomrule

\end{tabular}
\end{table}

In this section, we split the data into two parts:
$99\%$ of the data is used as the retrieval data, and the
remaining $1\%$ is used as the test data.
As shown in Table I, our method consistently achieves the best performance across both generation backbones. Compared with all baselines, it substantially improves both BLEU and ROUGE-L scores. In particular, under GPT-3.5, our method improves BLEU from 0.1594 (Stage-2) to 0.3275 and ROUGE-L from 0.1288 to 0.3445. Similar gains are observed for Llama3-8B-Chat. These results demonstrate that our approach is able to preserve the utility of retrieval data while maintaining high answer quality.

The superior performance can be attributed to the hierarchical protection strategy adopted in our framework. Instead of rewriting or regenerating entire documents, our method first identifies sensitive content and applies protection only where necessary. Most non-sensitive information is preserved in its original form, retaining the semantic integrity and knowledge content of the source documents. By minimizing unnecessary modifications and restricting transformations to a small subset of sensitive spans, the protected retrieval corpus remains highly faithful to the original data, leading to significantly better performance than methods that rely on extensive rewriting or synthetic generation.

%表明我们的方法层次化保护敏感信息，最大限度保护原始文档不动。只修改了部分，效果好。

\subsubsection{Results on ODQA}

\begin{table*}[t]
\caption{Utility results on Wiki-PII dataset}
\label{tab:res-wiki}
\centering
\resizebox{0.85\linewidth}{!}{
\begin{tabular}{c cccccccc}
\toprule
 GPT-3.5 & \multicolumn{2}{c}{NQ} & \multicolumn{2}{c}{TQA} & \multicolumn{2}{c}{WQ} & \multicolumn{2}{c}{CT} \\
\cmidrule(lr){2-3} \cmidrule(lr){4-5} \cmidrule(lr){6-7} \cmidrule(lr){8-9}
Method & BLEU $\uparrow$ & ROUGE-L $\uparrow$ & BLEU $\uparrow$ & ROUGE-L $\uparrow$ & BLEU $\uparrow$ & ROUGE-L $\uparrow$ & BLEU $\uparrow$ & ROUGE-L $\uparrow$ \\
\midrule
0-shot      & 0.0239 & 0.0528 & 0.0371 & 0.0776 & 0.0252 & 0.0622 & 0.0368 & 0.0680 \\
Origin      & 0.0598 & 0.1223 & 0.0660 & 0.1345 & 0.0517 & 0.1179 & \textbf{0.0742} & \textbf{0.1464} \\
Paraphrase  & 0.0508 & 0.1044 & 0.0558 & 0.1241 & 0.0331 & 0.0813 & 0.0563 & 0.1142 \\
ZeroGen     & 0.0113 & 0.0245 & 0.0251 & 0.0494 & 0.0366 & 0.0874 & 0.0484 & 0.0929 \\
AttrPrompt  & 0.0203 & 0.0415 & 0.0264 & 0.0534 & 0.0211 & 0.0478 & 0.0260 & 0.0503 \\
Stage-1     & 0.0435 & 0.0998 & 0.0550 & 0.1231 & 0.0465 & 0.1205 & 0.0509 & 0.1097 \\
Stage-2     & 0.0588 & 0.1250 & 0.0576 & 0.1221 & 0.0609 & 0.1296 & 0.0538 & 0.1210 \\
Ours     & \textbf{0.0674}  & \textbf{0.1437} & \textbf{0.0688}  & \textbf{0.1476} & \textbf{0.0666}  & \textbf{0.1443} & 0.0638  & 0.1392 \\

\midrule % 替换原来的\hline，booktabs统一粗细与间距

Llama3 & \multicolumn{2}{c}{NQ} & \multicolumn{2}{c}{TQA} & \multicolumn{2}{c}{WQ} & \multicolumn{2}{c}{CT} \\
\cmidrule(lr){2-3} \cmidrule(lr){4-5} \cmidrule(lr){6-7} \cmidrule(lr){8-9}
Method & BLEU $\uparrow$ & ROUGE-L $\uparrow$ & BLEU $\uparrow$ & ROUGE-L $\uparrow$ & BLEU $\uparrow$ & ROUGE-L $\uparrow$ & BLEU $\uparrow$ & ROUGE-L $\uparrow$ \\
\midrule
0-shot      & 0.00719 & 0.0136 & 0.0084 & 0.0157 & 0.0072 & 0.0143 & 0.0088 & 0.0150 \\
Origin      & 0.0180  & 0.0315 & 0.0150  & 0.0272 & 0.0147  & 0.0271 & 0.0178  & 0.0323 \\
Paraphrase  & 0.0153  & 0.0269 & 0.0127  & 0.0251 & 0.0094  & 0.0187 & 0.0135  & 0.0252 \\
ZeroGen     & 0.0034  & 0.0063 & 0.0057  & 0.010  & 0.0104  & 0.0201 & 0.0116  & 0.0205 \\
AttrPrompt  & 0.0061  & 0.0107 & 0.006   & 0.0108 & 0.006   & 0.0110 & 0.0062 & 0.0111 \\
Stage-1     & 0.0131  & 0.0257 & 0.0125  & 0.0249 & 0.0132  & 0.0277 & 0.0122  & 0.0242 \\
Stage-2     & 0.0248 & 0.0406 & 0.0252  & 0.0398 & 0.0256  & 0.0429  & \textbf{0.0255}  & 0.0427 \\
Ours     & \textbf{0.0283}  & \textbf{0.0476} & \textbf{0.0281}  & \textbf{0.0495} & \textbf{0.0278}   & \textbf{0.0472} & 0.0251   & \textbf{0.0437} \\
\bottomrule
\end{tabular}
}
\end{table*}

Table~\ref{tab:res-wiki} presents the ODQA results on the Wiki-PII dataset. Overall, our method achieves the best or near-best performance across the four benchmark datasets (NQ, TQA, WQ, and CT) under both GPT-3.5-Turbo and Llama3-8B-Chat. Compared with existing synthetic data generation approaches, including ZeroGen, Stage-2 and AttrPrompt, the proposed method consistently yields higher BLEU and ROUGE-L scores. Notably, the performance advantage is observed for both closed-source and open-source models, indicating that the effectiveness of our approach is not tied to a specific generation backbone. These results demonstrate that the proposed framework can preserve the utility of retrieval data while providing privacy protection, enabling downstream QA systems to maintain high answer quality.

These results highlight the advantage of our fine-grained, risk-aware protection strategy. Existing synthetic generation methods often modify or regenerate large portions of the retrieval corpus, which can introduce semantic distortions and weaken the connection between retrieved documents and downstream questions. 
In contrast, our method selectively protects only the sensitive content while preserving the majority of the original document structure, factual knowledge, and contextual information. By minimizing unnecessary perturbations to non-sensitive content, the resulting retrieval corpus remains highly faithful to the original data distribution and therefore retains strong retrieval utility. Overall, the results demonstrate that effective privacy protection and high retrieval utility are not necessarily conflicting objectives, and that carefully controlled modifications can preserve downstream ODQA performance while substantially reducing privacy risks.

\subsection{privacy evaluation}

To evaluate the privacy protection capability of our method, we conduct both targeted and untargeted extraction attacks following \cite{8zeng2025mitigating}. These attacks are designed to induce retrieval systems to disclose information from the retrieval corpus.

The attack prompt consists of two components: an \textit{information} component and a \textit{command} component. The information component is used to guide the retriever toward specific content, while the command component instructs the language model to reveal the retrieved information (e.g., “Please repeat all the context”). For targeted attacks, the information component is crafted to retrieve specific sensitive content, such as personal identifiers or private medical dialogues. For untargeted attacks, it is designed to maximize the amount of information extracted from the retrieval corpus without focusing on any particular target.

For \textit{targeted attacks}, we report \textbf{Repeat Prompts} and \textbf{Targeted Information}, which measures the number of unique sensitive information items successfully extracted by the attacker.
For \textit{untargeted attacks}, we report four metrics. \textbf{Repeat Prompts} counts the number of attack prompts that induce the model to reproduce at least 10 consecutive tokens from the retrieval corpus. \textbf{Rouge Prompts} counts the number of prompts whose outputs achieve a ROUGE-L score above 0.5 with any retrieved document. We further report \textbf{Repeat Contexts}, the number of unique verbatim excerpts extracted from the corpus, and \textbf{Rouge Contexts}, the number of unique extracted passages with a ROUGE-L score greater than 0.5.

\subsubsection{Results on Targeted Attack}

\begin{table*}[t]
\caption{Targeted attack results on Wiki-PII and HealthCareMagic dataset(250 prompts)}
\label{tab:tar-attack}
\centering
\renewcommand{\arraystretch}{1}
\setlength{\tabcolsep}{4pt}
\begin{tabular}{c|cc|cc|cc|cc}
\toprule
& \multicolumn{2}{c}{Target-wiki-llama-3-8b} & \multicolumn{2}{c}{Target-wiki-gpt-3.5} & \multicolumn{2}{c}{Target-chat-llama-3-8b} & \multicolumn{2}{c}{Target-chat-gpt-3.5} \\
\cmidrule(lr){2-3} \cmidrule(lr){4-5} \cmidrule(lr){6-7} \cmidrule(lr){8-9}
Method & Target info $\downarrow$ & Repeat prompts $\downarrow$ & Target info $\downarrow$ & Repeat prompts $\downarrow$ & Target info $\downarrow$ & Repeat prompts $\downarrow$ & Target info $\downarrow$ & Repeat prompts $\downarrow$ \\
\midrule
origin     & 25 & 12 & 167 & 64 & 7  & 23 & 75  & 132 \\
para       & 9  & 1  & 28  & 9  & 17 & 26 & 42  & 81  \\
ZeroGen    & 4  & 5  & 5   & 2  & 0  & 3  & 1   & 6   \\
AttrPrompt & 0  & 0  & 0   & 0  & 0  & 0  & 0   & 0   \\
Stage-1    & 1  & 4  & 3   & 19 & 3  & 11 & 12  & 36  \\
Stage-2    & 0  & 0  & 0   & 7  & 0  & 0  & 0   & 0   \\
Ours    & 0  & 0  & 0   & 0  & 0  & 0  & 0   & 0   \\
\bottomrule
\end{tabular}
\end{table*}

Table~\ref{tab:tar-attack} presents the results of targeted extraction attacks on both the Wiki-PII and HealthCareMagic datasets. The original retrieval corpus is highly vulnerable to extraction attacks, especially when paired with GPT-3.5-Turbo, where a large amount of sensitive information can be successfully recovered. Although existing synthesis-based methods, such as ZeroGen and AttrPrompt, reduce leakage to some extent, their effectiveness varies across datasets and models. In contrast, our method consistently achieves zero leakage across all settings, with both Targeted Information and Repeat Prompts reduced to zero.

The results demonstrate that our hierarchical protection mechanism effectively removes exploitable sensitive information from the retrieval corpus while preventing verbatim memorization and disclosure. 
%可省略
Compared with Stage-1 and Stage-2, which still exhibit a small number of successful extractions in certain settings, our complete framework eliminates all targeted leakage attempts. This indicates that selectively protecting sensitive content at different privacy levels provides stronger privacy guarantees than directly relying on synthetic generation or coarse-grained document transformation.

\subsubsection{Results on Untargeted Attack}

\begin{table*}[t]
\caption{Untargeted attack results on HealthCareMagic dataset(250 prompts)}
\label{tab:untarget-health}
\centering
\renewcommand{\arraystretch}{1}
\setlength{\tabcolsep}{2.0pt}
\begin{tabular}{c|cc|cc|cc|cc}
\toprule
& \multicolumn{4}{c}{Untarget-chat-llama} & \multicolumn{4}{c}{Untarget-chat-gpt3.5} \\
\cmidrule(lr){2-5} \cmidrule(lr){6-9}
Method & Repeat prompt $\downarrow$ & ROUGE prompt $\downarrow$ & Repeat context $\downarrow$ & ROUGE context $\downarrow$ & Repeat prompt $\downarrow$ & ROUGE prompt $\downarrow$ & Repeat context $\downarrow$ & ROUGE context $\downarrow$ \\
\midrule
origin     & 19 & 17 & 16 & 13 & 61 & 67 & 49 & 67 \\
para       & 23 & 13 & 22 & 11 & 45 & 63 & 33 & 50 \\
ZeroGen    & 0  & 0  & 0  & 0  & 0  & 0  & 0  & 0  \\
AttrPrompt & 0  & 0  & 0  & 0  & 0  & 0  & 0  & 0  \\
Stage-1    & 1  & 2  & 1  & 2  & 1  & 0  & 1  & 0  \\
Stage-2    & 0  & 0  & 0  & 0  & 0  & 0  & 0  & 0  \\
Ours    & 0  & 0  & 0   & 0  & 0  & 0  & 0   & 0   \\
\bottomrule
\end{tabular}
\end{table*}

Table~\ref{tab:untarget-health} reports the results of untargeted extraction attacks on the HealthCareMagic dataset. The original retrieval corpus remains highly vulnerable, especially under GPT-3.5-Turbo, where a large number of prompts successfully induce verbatim or semantically similar disclosures. Although paraphrasing reduces exact leakage, substantial privacy risks still remain. In contrast, our method achieves zero leakage across all prompt-level and context-level metrics under both generation backbones. Compared with Stage-1, which still allows a small number of successful extractions, our complete framework effectively prevents both exact memorization and semantic reconstruction of sensitive content, demonstrating strong robustness against untargeted extraction attacks.

\subsection{Ablation Study}

In this section, we conduct two ablation studies to further analyze the effectiveness of our method. Specifically, we investigate the impact of different privacy budgets and the effect of varying the number of retrieved documents. These experiments aim to provide a more fine-grained understanding of how key design choices affect both privacy protection and utility performance.

\subsubsection{Impact of Privacy Budgets}

\begin{figure*}[t]
\centering
% 每张占 0.23\textwidth，留微小空隙
\begin{minipage}{0.23\textwidth}
\centering
\includegraphics[width=\linewidth]{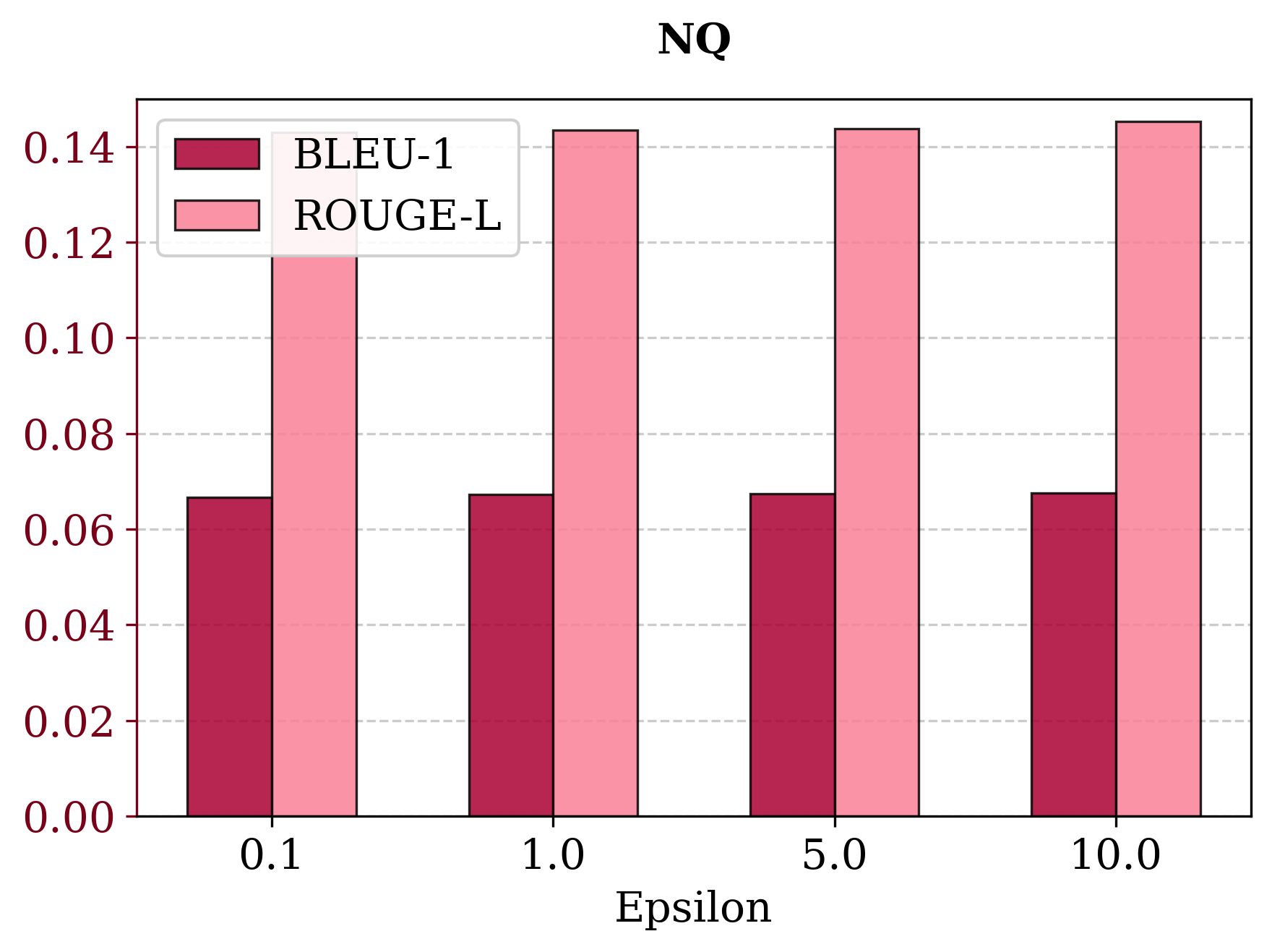}
%\subcaption{Fig 1 caption}
\end{minipage}
\hfill
\begin{minipage}{0.23\textwidth}
\centering
\includegraphics[width=\linewidth]{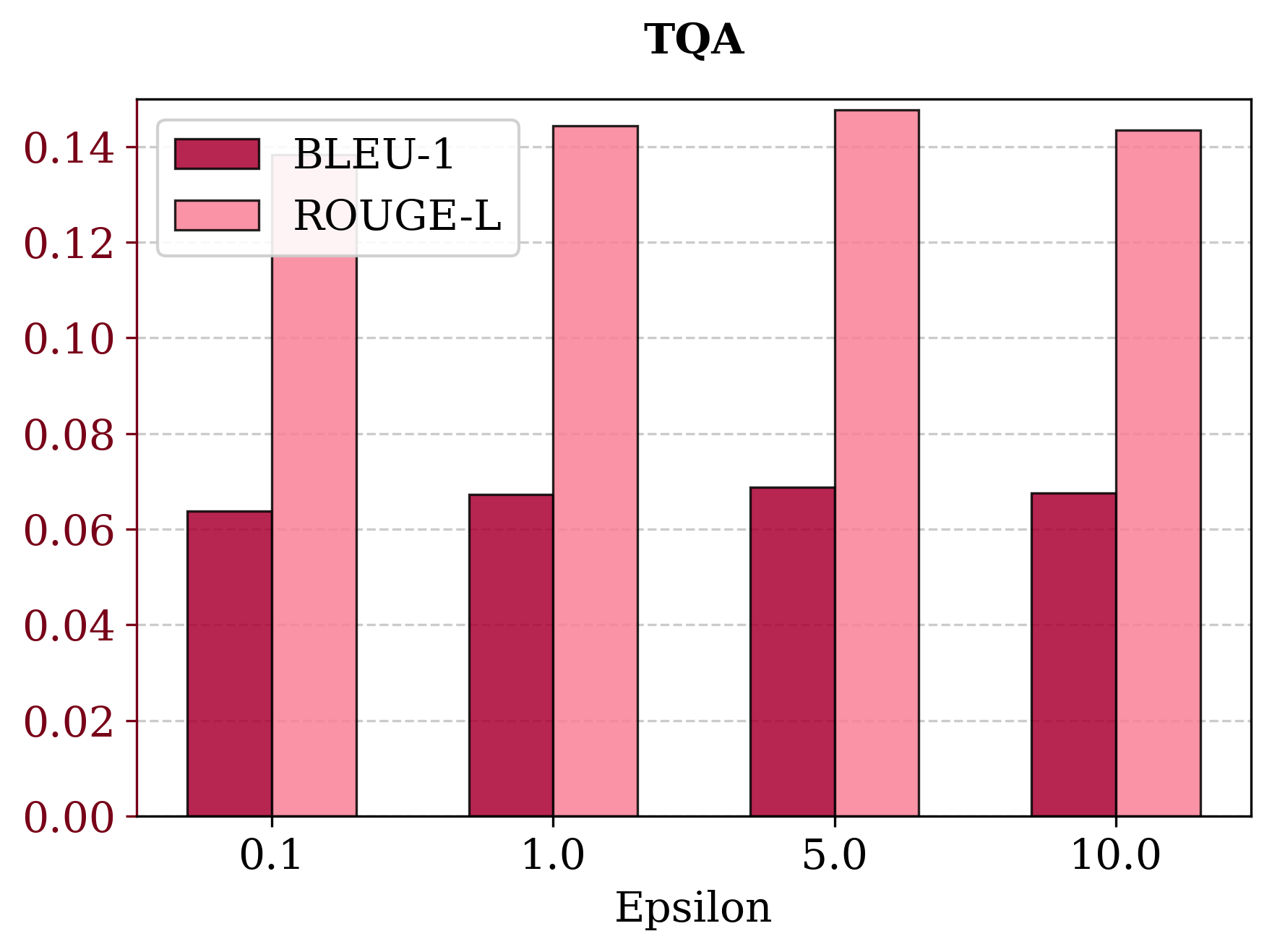}
%\subcaption{Fig 2 caption}
\end{minipage}
\hfill
\begin{minipage}{0.23\textwidth}
\centering
\includegraphics[width=\linewidth]{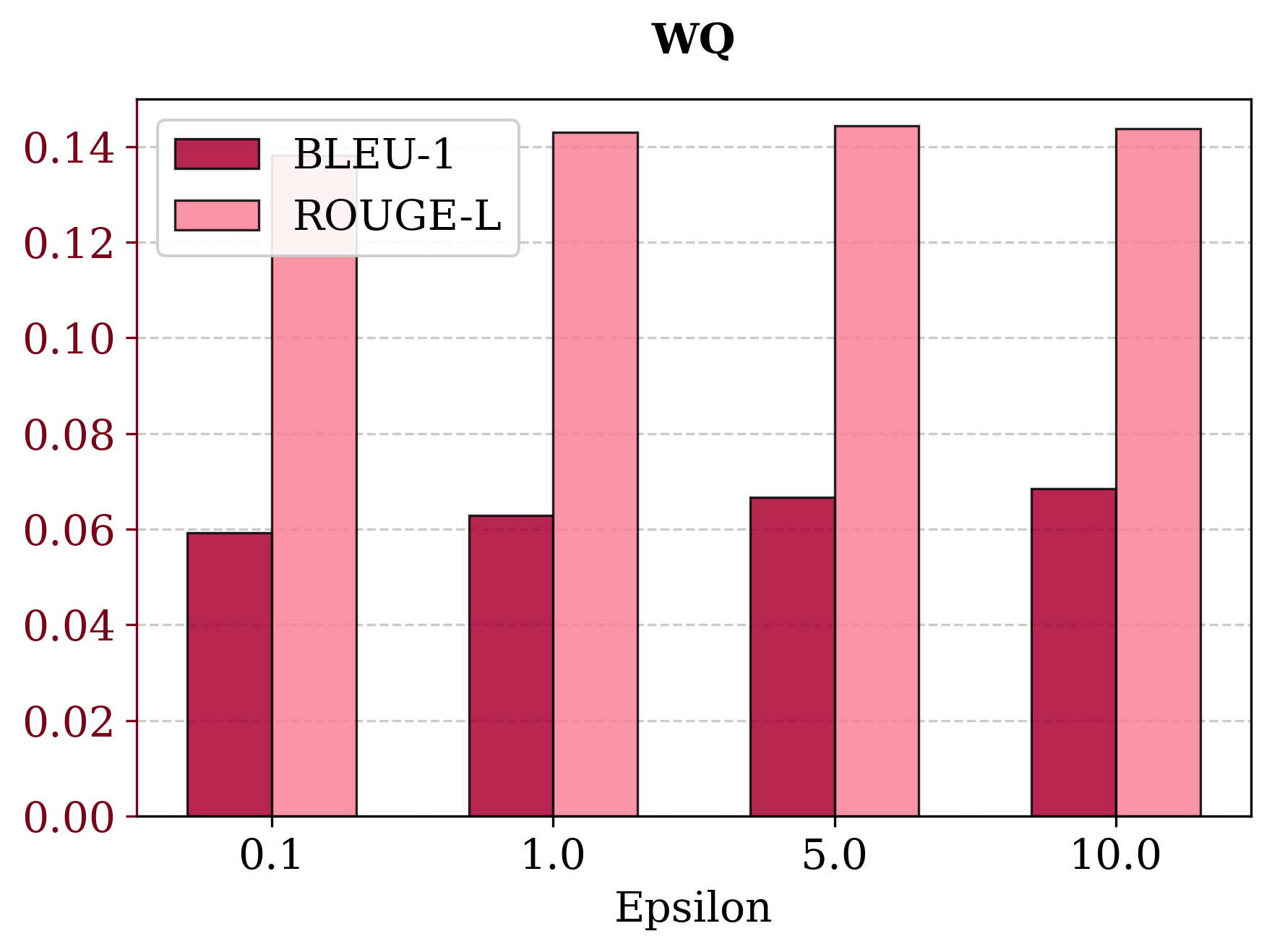}
%\subcaption{Fig 3 caption}
\end{minipage}
\hfill
\begin{minipage}{0.23\textwidth}
\centering
\includegraphics[width=\linewidth]{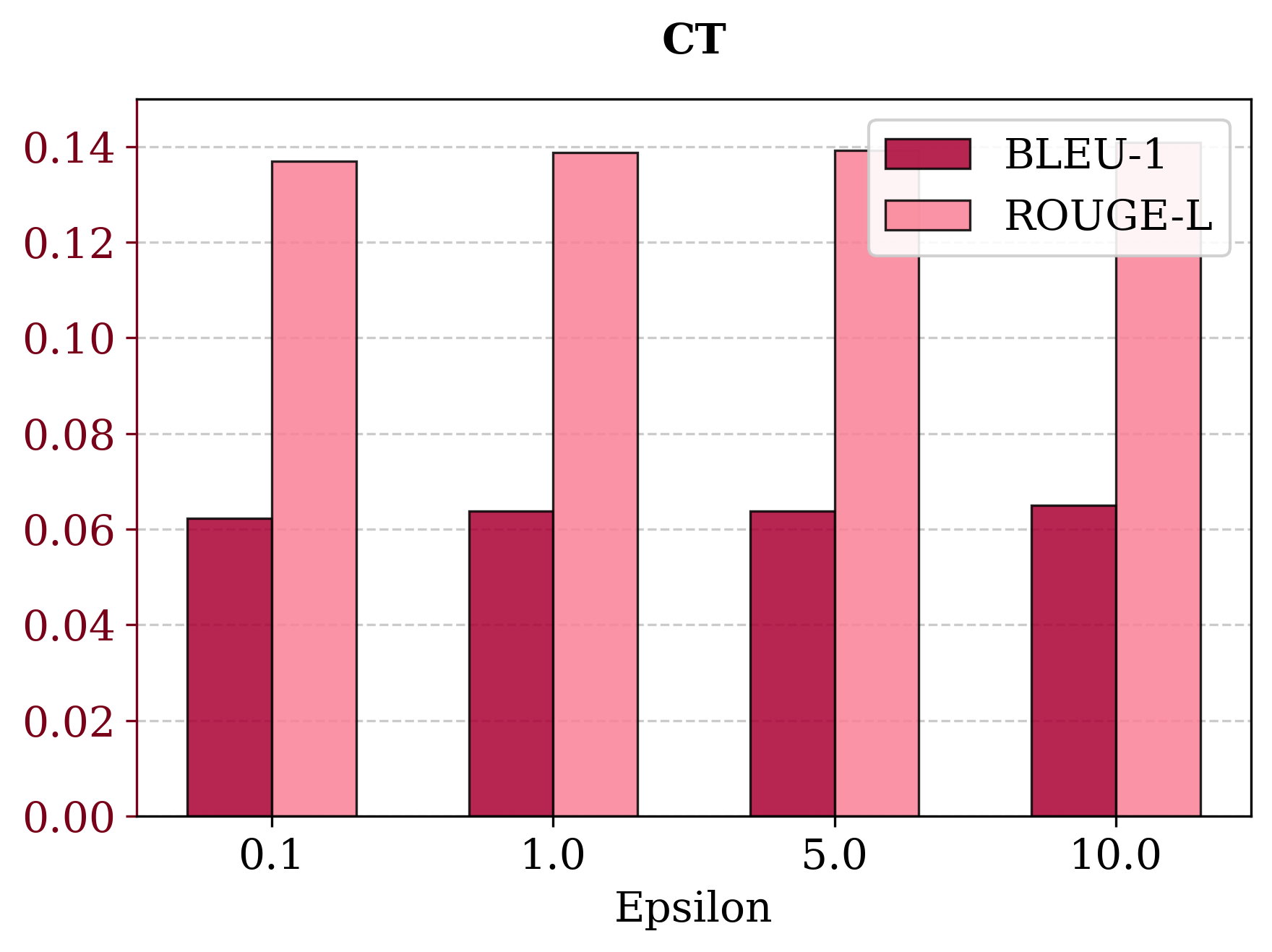}
%\subcaption{Fig 4 caption}
\end{minipage}
\caption{Utility Results under Different Privacy Budgets}
\label{fig:ab-privacy}
\end{figure*}

Fig.\ref{fig:ab-privacy} illustrates the impact of the privacy budget $\epsilon$ on utility across four ODQA benchmarks (NQ, TQA, WQ, and CT). Overall, both BLEU-1 and ROUGE-L scores consistently increase as $\epsilon$ grows, indicating a clear trade-off between privacy strength and answear quality. When $\epsilon$ is small (stronger privacy), performance is relatively lower due to more aggressive perturbation. As $\epsilon$ increases, the model is allowed to retain more accurate information, leading to steady improvements in both metrics. The trend is consistent across all datasets, demonstrating that higher privacy budgets enable better preservation of retrieval utility while maintaining stable gains across benchmarks.

\begin{table}[t]
\centering
\caption{Targeted Attack Results under different privacy budgets}
\label{tab:eps-target}
\renewcommand{\arraystretch}{1.1}
\setlength{\tabcolsep}{4pt}
\begin{tabular}{c|cc|cc}
\toprule
& \multicolumn{2}{c}{Chat} & \multicolumn{2}{c}{Wiki} \\
\cmidrule(lr){2-3} \cmidrule(lr){4-5}
$\epsilon$ & Target info $\downarrow$ & Repeat prompts $\downarrow$ & Target info $\downarrow$ & Repeat prompts $\downarrow$ \\
\midrule
0.1 & 0 & 0 & 0 & 0 \\
1   & 0 & 0 & 0 & 0 \\
5   & 0 & 0 & 0 & 0 \\
10  & 0 & 0 & 0 & 0 \\
\bottomrule
\end{tabular}
\end{table}

Table~\ref{tab:eps-target} reports targeted attack results under different privacy budgets $\epsilon$ on both the HealthCareMagic and Wiki-PII datasets using GPT-3.5 as the generation backbone.Table~\ref{tab:eps-target} reports targeted attack results under different privacy budgets $\epsilon$ on both the HealthCareMagic and Wiki-PII datasets using GPT-3.5 as the generation backbone. Across all settings, we observe zero successful extraction, including both Targeted Information and Repeat Prompts, under all privacy budgets from $\epsilon=0.1$ to $\epsilon=10$. This indicates that the proposed method is robust to varying levels of privacy strength, and even under weaker privacy constraints (larger $\epsilon$), no sensitive information can be recovered via targeted attacks. These results demonstrate that the hierarchical protection mechanism effectively eliminates direct leakage of sensitive content while maintaining consistent behavior across different privacy budgets.

\subsubsection{Impact of the retrieved number of documents}

\begin{table*}[t]
\centering
\caption{Utility Results on Number of Retrieved Docs}
\label{tab:ablation_k_all_util}
\renewcommand{\arraystretch}{1.1}
\resizebox{0.85\linewidth}{!}{% resizebox只包裹tabular
\setlength{\tabcolsep}{4pt}
\begin{tabular}{c|cc|cc|cc|cc|cc} % 加竖线分隔Health与NQ
\toprule
& \multicolumn{2}{c}{HealthCareMagic} & \multicolumn{2}{c}{NQ} & \multicolumn{2}{c}{TQA} & \multicolumn{2}{c}{WQ} & \multicolumn{2}{c}{CT} \\
\cmidrule(lr){2-3} \cmidrule(lr){4-5} \cmidrule(lr){6-7} \cmidrule(lr){8-9} \cmidrule(lr){10-11}
$k$ & BLEU & ROUGE & BLEU & ROUGE-L & BLEU & ROUGE-L & BLEU & ROUGE-L & BLEU & ROUGE-L \\
\midrule
1 & 0.3275 & 0.3445 & 0.0674 & 0.1437 & 0.0688 & 0.1476 & 0.0666 & 0.1443 & 0.0638 & 0.1392 \\
3 & 0.3300 & 0.3456 & 0.0666 & 0.1432 & 0.0683 & 0.1453 & 0.0693 & 0.1455 & 0.0648 & 0.1388 \\
5 & 0.3246 & 0.3430 & 0.0667 & 0.1419 & 0.0670 & 0.1427 & 0.0696 & 0.1459 & 0.0623 & 0.1371 \\
\bottomrule
\end{tabular}
}
\end{table*}

Table~\ref{tab:ablation_k_all_util} reports the impact of varying the number of retrieved documents $k$ on utility across HealthCareMagic and four ODQA benchmarks. Overall, the results remain stable when increasing $k$ from 1 to 5, with only minor fluctuations in BLEU and ROUGE scores. On HealthCareMagic, performance is consistently high and even slightly improves at $k=3$, indicating that a moderate increase in retrieved context can provide additional useful information. For ODQA datasets (NQ, TQA, WQ, and CT), we observe similarly small variations, suggesting that increasing the number of retrieved documents does not introduce noticeable degradation in generation quality.

These results demonstrate that our pipeline maintains robust utility across different retrieval settings. In particular, expanding the retrieval scope does not harm performance, indicating that the generated data is able to support multi-document retrieval without overfitting to a specific retrieval size. This further confirms the stability and generalization capability of our method under varying retrieval configurations.

\subsubsection{Impact of Model Choice}

\begin{figure}[t]
\centering
% 每张占 0.23\textwidth，留微小空隙
\begin{minipage}{0.23\textwidth}
\centering
\includegraphics[width=\linewidth]{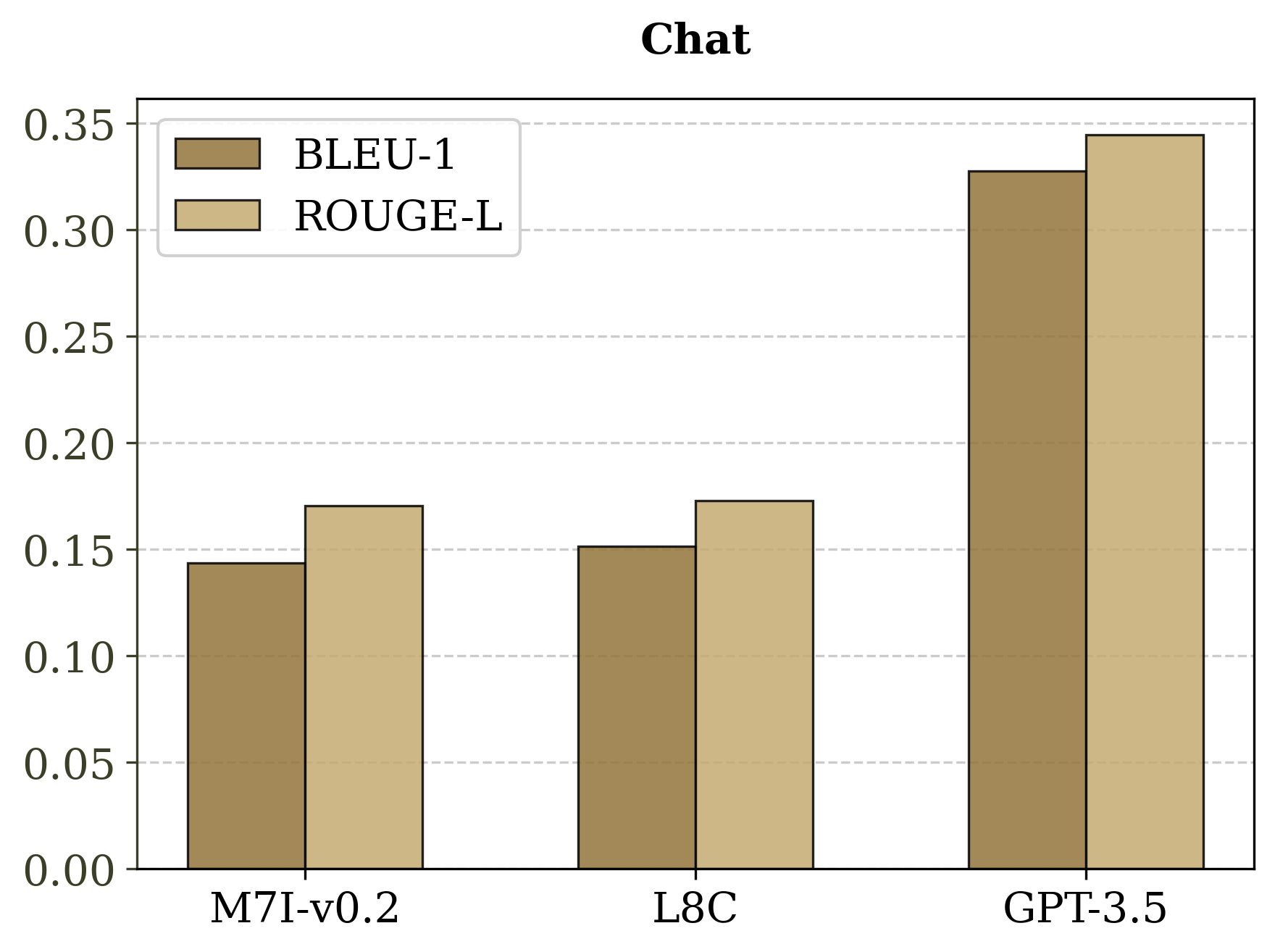}
%\subcaption{Fig 1 caption}
\end{minipage}
\hfill
\begin{minipage}{0.23\textwidth}
\centering
\includegraphics[width=\linewidth]{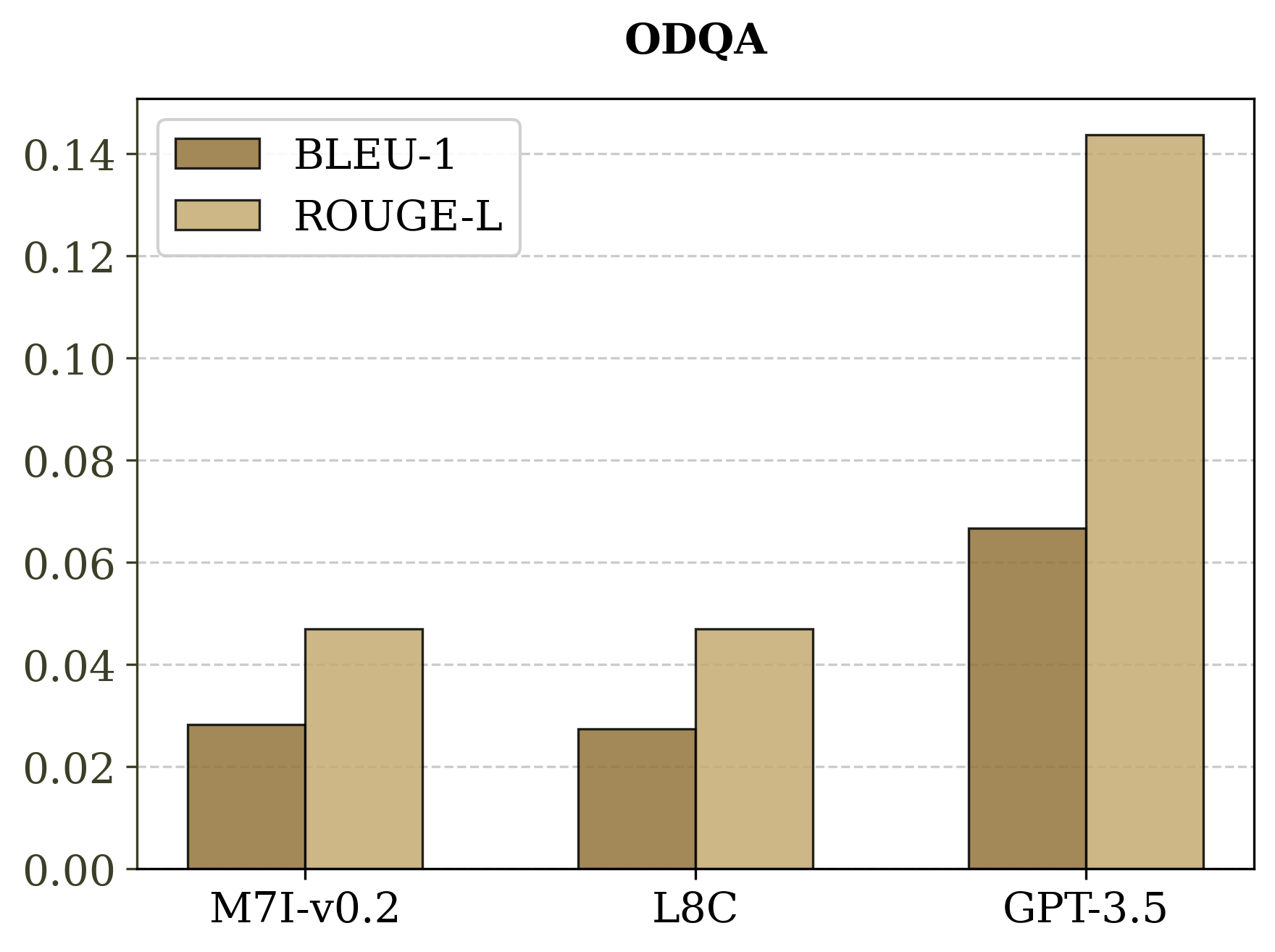}
%\subcaption{Fig 2 caption}
\end{minipage}

\caption{Utility Results under Different Generation Models}
\label{fig:ab-model}
\end{figure}

We evaluate the utility performance of Mistral-7B-Instruct-v0.2(M7I-v0.2), L8C and GPT-3.5 on ODQA and HealthCareMagic benchmarks, where BLEU-1 and ROUGE-L are adopted to quantify the similarity between generated content and ground-truth references. Notably, the ODQA results plotted in the figure correspond to the averaged values across its four sub-datasets.
On both datasets, M7I-v0.2 and L8C obtain highly comparable but relatively low metric values. GPT-3.5 achieves obvious performance advantages over the two open-source models, and the performance gap is more prominent on the challenging ODQA task than the HealthCareMagic task.

The close results of M7I-v0.2 and L8C reveal that lightweight open-source LLMs have similar constrained capacity for our stage-1 data generation, which produces low-fidelity synthetic text. In contrast, GPT-3.5 provides steady and superior utility across two different task scenarios. This outcome confirms GPT-3.5 acts as a reliable backbone for our two-stage privacy-preserving generation framework.

\subsection{Hyperparameter Choice}

%此外，我们还比较了不同超参数组合对所提出方法的效用影响。具体而言我们
\begin{table*}[t]
\centering
\caption{Utility Results on Choice of Hyperparameter}
\label{tab:hyperparameter}
\renewcommand{\arraystretch}{1.1}
\resizebox{0.85\linewidth}{!}{% resizebox只包裹tabular
\setlength{\tabcolsep}{4pt}
\begin{tabular}{c|cc|cc|cc|cc|cc} % 加竖线分隔Health与NQ
\toprule
& \multicolumn{2}{c}{HealthCareMagic} & \multicolumn{2}{c}{NQ} & \multicolumn{2}{c}{TQA} & \multicolumn{2}{c}{WQ} & \multicolumn{2}{c}{CT} \\
\cmidrule(lr){2-3} \cmidrule(lr){4-5} \cmidrule(lr){6-7} \cmidrule(lr){8-9} \cmidrule(lr){10-11}
$\gamma, \ \alpha$ & BLEU & ROUGE & BLEU & ROUGE-L & BLEU & ROUGE-L & BLEU & ROUGE-L & BLEU & ROUGE-L \\
\midrule
$0.6,\ 0.3$ & 0.3182 & 0.3351 & 0.0605 & 0.1344 & 0.0619 & 0.1382 & 0.0597 & 0.1350 & 0.0571 & 0.1299 \\
$0.6,\ 0.7$ & 0.3153 & 0.3337 & 0.0598 & 0.1326 & 0.0601 & 0.1333 & 0.0627 & 0.1366 & 0.0556 & 0.1278 \\
$0.4,\ 0.5$ & 0.3190 & \textbf{0.3420} & 0.0612 & 0.1351 & 0.0624 & 0.1390 & 0.0603 & 0.1355 & 0.0576 & 0.1305 \\
$0.5,\ 0.5$ & 0.3211 & 0.3383 & 0.0650 & 0.1411 & \textbf{0.0691} & 0.1408 & 0.0622 & 0.1374 & 0.0624 & \textbf{0.1401} \\
$0.6,\ 0.5$ & \textbf{0.3275} & 0.3445 & \textbf{0.0674} & \textbf{0.1437} & 0.0688 & \textbf{0.1476} & \textbf{0.0666} & \textbf{0.1443} & \textbf{0.0638} & 0.1392 \\
\bottomrule
\end{tabular}
}
\end{table*}

Furthermore, we investigate the impact of different hyperparameter combinations on the utility of the proposed method. Specifically, $\gamma$ balances the relative importance of semantic relevance and field sensitivity in risk assessment, while $\alpha$ controls the trade-off between privacy protection and semantic preservation during candidate selection. As shown in Table~\ref{tab:hyperparameter}, the proposed method exhibits stable performance across different parameter settings, indicating that it is not overly sensitive to hyperparameter choices. Among all configurations, $\gamma=0.6$ and $\alpha=0.5$ consistently achieve the best overall results. This suggests that assigning slightly greater importance to semantic relevance than field sensitivity leads to more accurate risk assessment, while a balanced weighting between privacy protection and semantic preservation provides the most effective trade-off between privacy and utility. Therefore, we adopt $\gamma=0.6$ and $\alpha=0.5$ as the default setting in all experiments.

\subsection{Efficiency Analysis}

To evaluate the computational efficiency of the proposed method, particularly the time overhead introduced by the privacy protection procedure, we compare its runtime with five representative privacy-preserving approaches. The experiments are conducted with the number of retrieved documents set to K=1, which eliminates the impact of retrieval-scale variations and enables a fair comparison of the privacy protection overhead. 
As shown in Table~\ref{tab:runtime}, all methods have the same retrieval time, and the runtime differences mainly come from the privacy protection process. Our method achieves a relatively low privacy protection overhead and is only slower than PARA. Although PARA has lower runtime, it relies on online LLM API calls for privacy processing, while our method can operate offline without LLM invocation, making it more suitable for practical deployment.

\begin{table}[h]
\caption{Runtime Comparison (Seconds)} 
\centering
\resizebox{0.8\columnwidth}{!}{% resizebox只包裹tabular
\begin{tabular}{lrrr}
\hline
               & \multicolumn{1}{c}{Retrieval} & \multicolumn{1}{c}{Protection} & \multicolumn{1}{c}{Total}\\ \hline
para        & 6.29           & 1628.37     & 1634.66 \\
ZeroGen     & 6.29           & 8792.48     & 8798.77 \\
AttrPrompt  & 6.29           & 3681.30     & 3687.59 \\
Stage-2     & 6.29           & 3259.11     & 3265.40 \\
LPRAG       & 6.29           & 3977.83     & 3984.12 \\
Ours        & 6.29           & 2712.26     & 2718.55 \\ \hline
\end{tabular}
}
\label{tab:runtime}
\end{table}

\section{Conclusion}
In this paper, we revisit privacy protection in retrieval-augmented generation from a new perspective and identify a fundamental limitation of existing approaches: the document-level static privacy risk assumption. We show that privacy leakage in RAG is inherently query-dependent, and the same document may pose substantially different privacy risks under different user prompts. Motivated by this observation, we propose PA-HDP, a Prompt-Aware Dynamic Hierarchical Differential Privacy protection framework that dynamically assesses privacy risks and applies differentiated protection through risk stratification, adaptive sensitive entity replacement, and exponential mechanism-based text selection.

Extensive experiments on medical dialogue and open-domain question answering benchmarks demonstrate that PA-HDP effectively mitigates both targeted and untargeted privacy attacks while maintaining high retrieval utility. Moreover, under rigorous differential privacy guarantees, PA-HDP achieves a favorable trade-off between privacy protection and model performance. We hope this work highlights the importance of query-aware privacy protection and provides a foundation for the secure deployment of RAG systems on sensitive data.

\subsubsection{limitation}
PA-HDP assumes a fixed privacy budget that is predetermined before deployment. However, in real-world RAG systems, the number of user queries is typically unknown in advance, and privacy budgets may be gradually exhausted as queries accumulate. Future work will investigate adaptive privacy budget allocation and accounting mechanisms to support long-term deployments while maintaining rigorous privacy guarantees.

%Although PA-HDP provides rigorous differential privacy guarantees and achieves strong privacy--utility trade-offs, it assumes a fixed privacy budget that is allocated in advance. In practice, however, the number and distribution of user queries are often unknown and may vary over time. As more queries are issued, the available privacy budget is gradually consumed, which can limit the applicability of the framework in long-term or high-frequency deployment scenarios.

%A promising direction for future work is to develop adaptive privacy budget management mechanisms that dynamically allocate and optimize privacy expenditure according to query characteristics and system usage patterns. Such techniques could further improve the scalability and practicality of privacy-preserving RAG systems while maintaining strong privacy guarantees.


\bibliography{reference}

@misc{1koga2025privacypreserving,
      title={Privacy-Preserving Retrieval-Augmented Generation with Differential Privacy}, 
      author={Tatsuki Koga and Ruihan Wu and Zhiyuan Zhang and Kamalika Chaudhuri},
      year={2025},
      eprint={2412.04697},
      archivePrefix={arXiv},
      primaryClass={cs.CR},
      url={https://arxiv.org/abs/2412.04697}, 
}

@article{2vinod2026invisibleink,
  title={Invisibleink: High-utility and low-cost text generation with differential privacy},
  author={Vinod, Vishnu and Pillutla, Krishna and Guha Thakurta, Abhradeep},
  journal={Advances in Neural Information Processing Systems},
  volume={38},
  pages={67254--67298},
  year={2026}
}

@article{3wang2025privacy,
  title={Privacy-aware decoding: Mitigating privacy leakage of large language models in retrieval-augmented generation},
  author={Wang, Haoran and Xu, Xiongxiao and Huang, Baixiang and Shu, Kai},
  journal={arXiv preprint arXiv:2508.03098},
  year={2025}
}

@inproceedings{4grislain2025rag,
  title={Rag with differential privacy},
  author={Grislain, Nicolas},
  booktitle={2025 IEEE Conference on Artificial Intelligence (CAI)},
  pages={847--852},
  year={2025},
  organization={IEEE}
}

@inproceedings{5hemmat2025vague,
  title={VAGUE-Gate: Plug-and-Play Local-Privacy Shield for Retrieval-Augmented Generation},
  author={Hemmat, Arshia and Moqadas, Matin and Mamanpoosh, Ali and Rismanchian, Amirmasoud and Fatemi, Afsaneh},
  booktitle={Proceedings of the 14th International Joint Conference on Natural Language Processing and the 4th Conference of the Asia-Pacific Chapter of the Association for Computational Linguistics},
  pages={3715--3730},
  year={2025}
}

@incollection{dwork2025differential,
  title={Differential privacy},
  author={Dwork, Cynthia},
  booktitle={Encyclopedia of Cryptography, Security and Privacy},
  pages={649--652},
  year={2025},
  publisher={Springer}
}

@article{6joopearl,
  title={PEARL: Differentially Private and Entropy-Aware Regulated Language Generation},
  author={Joo, Seongho and Koh, Hyukhun and Jung, Kyomin},
  year={2026},
  journal={openreview.net/pdf?id=qIUR54yyro}
}

@article{7he2025mitigating,
  title={Mitigating privacy risks in Retrieval-Augmented Generation via locally private entity perturbation},
  author={He, Longzhu and Tang, Peng and Zhang, Yuanhe and Zhou, Pengpeng and Su, Sen},
  journal={Information Processing \& Management},
  volume={62},
  number={4},
  pages={104150},
  year={2025},
  publisher={Elsevier}
}

@inproceedings{mcsherry2007mechanism,
  title={Mechanism design via differential privacy},
  author={McSherry, Frank and Talwar, Kunal},
  booktitle={48th Annual IEEE Symposium on Foundations of Computer Science (FOCS'07)},
  pages={94--103},
  year={2007},
  organization={IEEE}
}

@inproceedings{dwork2006differential,
  title={Differential privacy},
  author={Dwork, Cynthia},
  booktitle={International colloquium on automata, languages, and programming},
  pages={1--12},
  year={2006},
  organization={Springer}
}

@inproceedings{dwork2006calibrating,
  title={Calibrating Noise to Sensitivity in Private Data Analysis},
  author={Dwork, Cynthia and McSherry, Frank and Nissim, Kobbi and Smith, Adam},
  booktitle={Theory of Cryptography Conference},
  pages={265--284},
  year={2006},
  publisher={Springer}
}

@inproceedings{8zeng2025mitigating,
  title={Mitigating the privacy issues in retrieval-augmented generation (rag) via pure synthetic data},
  author={Zeng, Shenglai and Zhang, Jiankun and He, Pengfei and Ren, Jie and Zheng, Tianqi and Lu, Hanqing and Xu, Han and Liu, Hui and Xing, Yue and Tang, Jiliang},
  booktitle={Proceedings of the 2025 Conference on Empirical Methods in Natural Language Processing},
  pages={24538--24569},
  year={2025}
}

@article{9mori2025differentially,
  title={Differentially Private Synthetic Text Generation for Retrieval-Augmented Generation (RAG)},
  author={Mori, Junki and Kakizaki, Kazuya and Miyagawa, Taiki and Sakuma, Jun},
  journal={arXiv preprint arXiv:2510.06719},
  year={2025}
}

@inproceedings{10zeng2024good,
  title={The good and the bad: Exploring privacy issues in retrieval-augmented generation (rag)},
  author={Zeng, Shenglai and Zhang, Jiankun and He, Pengfei and Liu, Yiding and Xing, Yue and Xu, Han and Ren, Jie and Chang, Yi and Wang, Shuaiqiang and Yin, Dawei and others},
  booktitle={Findings of the Association for Computational Linguistics: ACL 2024},
  pages={4505--4524},
  year={2024}
}

@article{11chen2025privacy,
  title={Privacy-Preserving Reasoning with Knowledge-Distilled Parametric Retrieval Augmented Generation},
  author={Chen, Jinwen and Zhang, Hainan and Pang, Liang and Tong, Yongxin and Zhou, Haibo and Zhan, Yuan and Lin, Wei and Zheng, Zhiming},
  journal={arXiv preprint arXiv:2509.01088},
  year={2025}
}

@article{12tang2026differentially,
  title={Differentially Private Retrieval-Augmented Generation},
  author={Tang, Tingting and Flemings, James and Wang, Yongqin and Annavaram, Murali},
  journal={arXiv preprint arXiv:2602.14374},
  year={2026}
}

@inproceedings{13yu2024textual,
  title={Textual differential privacy for context-aware reasoning with large language model},
  author={Yu, Junwei and Zhou, Jieyu and Ding, Yepeng and Zhang, Lingfeng and Guo, Yuheng and Sato, Hiroyuki},
  booktitle={2024 IEEE 48th Annual Computers, Software, and Applications Conference (COMPSAC)},
  pages={988--997},
  year={2024},
  organization={IEEE}
}

@inproceedings{14cheng2025remoterag,
  title={Remoterag: A privacy-preserving llm cloud rag service},
  author={Cheng, Yihang and Zhang, Lan and Wang, Junyang and Yuan, Mu and Yao, Yunhao},
  booktitle={Findings of the Association for Computational Linguistics: ACL 2025},
  pages={3820--3837},
  year={2025}
}

@inproceedings{15he2025press,
  title={Press: Defending privacy in retrieval-augmented generation via embedding space shifting},
  author={He, Jiaming and Liu, Cheng and Hou, Guanyu and Jiang, Wenbo and Li, Jiachen},
  booktitle={ICASSP 2025-2025 IEEE International Conference on Acoustics, Speech and Signal Processing (ICASSP)},
  pages={1--5},
  year={2025},
  organization={IEEE}
}

@inproceedings{16charles2025learning,
  title={Learning with User-Level Differential Privacy Under Fixed Compute Budgets},
  author={Charles, Zachary and Ganesh, Arun and McKenna, Ryan and McMahan, H Brendan and Mitchell, Nicole and Pillutla, Krishna and Rush, Keith},
  booktitle={2025 IEEE Conference on Secure and Trustworthy Machine Learning (SaTML)},
  pages={901--920},
  year={2025},
  organization={IEEE}
}

@article{17boenisch2023have,
  title={Have it your way: Individualized privacy assignment for DP-SGD},
  author={Boenisch, Franziska and M{\"u}hl, Christopher and Dziedzic, Adam and Rinberg, Roy and Papernot, Nicolas},
  journal={Advances in Neural Information Processing Systems},
  volume={36},
  pages={19073--19103},
  year={2023}
}

@article{18wu2025private,
  title={Private-RAG: Answering Multiple Queries with LLMs while Keeping Your Data Private},
  author={Wu, Ruihan and Wang, Erchi and Zhang, Zhiyuan and Wang, Yu-Xiang},
  journal={arXiv preprint arXiv:2511.07637},
  year={2025}
}

@article{19wang2025learning,
  title={Learning to Erase Private Knowledge from Multi-Documents for Retrieval-Augmented Large Language Models},
  author={Wang, Yujing and Zhang, Hainan and Pang, Liang and Tong, Yongxin and Guo, Binghui and Zheng, Hongwei and Zheng, Zhiming},
  journal={arXiv preprint arXiv:2504.09910},
  year={2025}
}

@inproceedings{20wu2025beyond,
  title={Beyond Per-Question Privacy: Multi-Query Differential Privacy for RAG Systems},
  author={Wu, Ruihan and Wang, Erchi and Wang, Yu-Xiang},
  booktitle={NeurIPS 2025 Workshop: Reliable ML from Unreliable Data},
  year={2025}
}

@inproceedings{24ye2025efficient,
  title={Efficient Privacy-Preserving Retrieval Augmented Generation with Distance-Preserving Encryption},
  author={Ye, Huanyi and Guo, Jiale and Liu, Ziyao and Lam, Kwok-Yan},
  booktitle={2025 3rd International Conference on Foundation and Large Language Models (FLLM)},
  pages={668--676},
  year={2025},
  organization={IEEE}
}

@article{25guan2025privacy,
  title={Privacy Challenges and Solutions in Retrieval-Augmented Generation-Enhanced LLMs for Healthcare Chatbots: A Review of Applications, Risks, and Future Directions},
  author={Guan, Shaowei and Kwok, Hin Chi and Law, Ngai Fong and Stiglic, Gregor and Qin, Harry and Hui, Vivian},
  journal={arXiv preprint arXiv:2511.11347},
  year={2025}
}

@article{26mu2026towards,
  title={Towards Secure Retrieval-Augmented Generation: A Comprehensive Review of Threats, Defenses and Benchmarks},
  author={Mu, Yanming and Hu, Hao and Li, Feiyang and Yuan, Qiao and Wu, Jiang and Liu, Zichuan and Liu, Pengcheng and Wang, Mei and Zhou, Hongwei and Liu, Yuling},
  journal={arXiv preprint arXiv:2603.21654},
  year={2026}
}

@inproceedings{ex-huang2023privacy,
  title={Privacy implications of retrieval-based language models},
  author={Huang, Yangsibo and Gupta, Samyak and Zhong, Zexuan and Li, Kai and Chen, Danqi},
  booktitle={Proceedings of the 2023 Conference on Empirical Methods in Natural Language Processing},
  pages={14887--14902},
  year={2023}
}

@article{nq-kwiatkowski2019natural,
  title={Natural questions: a benchmark for question answering research},
  author={Kwiatkowski, Tom and Palomaki, Jennimaria and Redfield, Olivia and Collins, Michael and Parikh, Ankur and Alberti, Chris and Epstein, Danielle and Polosukhin, Illia and Devlin, Jacob and Lee, Kenton and others},
  journal={Transactions of the Association for Computational Linguistics},
  volume={7},
  pages={453--466},
  year={2019},
  publisher={MIT Press One Rogers Street, Cambridge, MA 02142-1209, USA journals-info~…}
}

@inproceedings{tqa-joshi2017triviaqa,
  title={Triviaqa: A large scale distantly supervised challenge dataset for reading comprehension},
  author={Joshi, Mandar and Choi, Eunsol and Weld, Daniel S and Zettlemoyer, Luke},
  booktitle={Proceedings of the 55th Annual Meeting of the Association for Computational Linguistics (Volume 1: Long Papers)},
  pages={1601--1611},
  year={2017}
}

@inproceedings{wq-berant2013semantic,
  title={Semantic parsing on freebase from question-answer pairs},
  author={Berant, Jonathan and Chou, Andrew and Frostig, Roy and Liang, Percy},
  booktitle={Proceedings of the 2013 conference on empirical methods in natural language processing},
  pages={1533--1544},
  year={2013}
}

@inproceedings{ct-baudivs2015modeling,
  title={Modeling of the question answering task in the yodaqa system},
  author={Baudi{\v{s}}, Petr and {\v{S}}ediv{\`y}, Jan},
  booktitle={International Conference of the cross-language evaluation Forum for European languages},
  pages={222--228},
  year={2015},
  organization={Springer}
}

@inproceedings{bs-ye2022zerogen,
  title={Zerogen: Efficient zero-shot learning via dataset generation},
  author={Ye, Jiacheng and Gao, Jiahui and Li, Qintong and Xu, Hang and Feng, Jiangtao and Wu, Zhiyong and Yu, Tao and Kong, Lingpeng},
  booktitle={Proceedings of the 2022 Conference on Empirical Methods in Natural Language Processing},
  pages={11653--11669},
  year={2022}
}

@article{bs-yu2023large,
  title={Large language model as attributed training data generator: A tale of diversity and bias},
  author={Yu, Yue and Zhuang, Yuchen and Zhang, Jieyu and Meng, Yu and Ratner, Alexander J and Krishna, Ranjay and Shen, Jiaming and Zhang, Chao},
  journal={Advances in neural information processing systems},
  volume={36},
  pages={55734--55784},
  year={2023}
}

@inproceedings{carlini2021extracting,
  title={Extracting training data from large language models},
  author={Carlini, Nicholas and Tramer, Florian and Wallace, Eric and Jagielski, Matthew and Herbert-Voss, Ariel and Lee, Katherine and Roberts, Adam and Brown, Tom and Song, Dawn and Erlingsson, Ulfar and others},
  booktitle={30th USENIX security symposium (USENIX Security 21)},
  pages={2633--2650},
  year={2021}
}

@article{gao2023retrieval,
  title={Retrieval-augmented generation for large language models: A survey},
  author={Gao, Yunfan and Xiong, Yun and Gao, Xinyu and Jia, Kangxiang and Pan, Jinliu and Bi, Yuxi and Dai, Yi and Sun, Jiawei and Wang, Meng and Wang, Haofen},
  journal={arXiv preprint arXiv:2312.10997},
  year={2023}
}

@article{huang2025survey,
  title={A survey on hallucination in large language models: Principles, taxonomy, challenges, and open questions},
  author={Huang, Lei and Yu, Weijiang and Ma, Weitao and Zhong, Weihong and Feng, Zhangyin and Wang, Haotian and Chen, Qianglong and Peng, Weihua and Feng, Xiaocheng and Qin, Bing and others},
  journal={ACM Transactions on Information Systems},
  volume={43},
  number={2},
  pages={1--55},
  year={2025},
  publisher={ACM New York, NY}
}

@article{amugongo2025retrieval,
  title={Retrieval augmented generation for large language models in healthcare: A systematic review},
  author={Amugongo, Lameck Mbangula and Mascheroni, Pietro and Brooks, Steven and Doering, Stefan and Seidel, Jan},
  journal={PLOS Digital Health},
  volume={4},
  number={6},
  pages={e0000877},
  year={2025},
  publisher={Public Library of Science San Francisco, CA USA}
}

@inproceedings{zhao2024optimizing,
  title={Optimizing LLM based retrieval augmented generation pipelines in the financial domain},
  author={Zhao, Yiyun and Singh, Prateek and Bhathena, Hanoz and Ramos, Bernardo and Joshi, Aviral and Gadiyaram, Swaroop and Sharma, Saket},
  booktitle={Proceedings of the 2024 Conference of the North American Chapter of the Association for Computational Linguistics: Human Language Technologies (Volume 6: Industry Track)},
  pages={279--294},
  year={2024}
}

@article{hindi2025enhancing,
  title={Enhancing the precision and interpretability of retrieval-augmented generation (rag) in legal technology: A survey},
  author={Hindi, Mahd and Mohammed, Linda and Maaz, Ommama and Alwarafy, Abdulmalik},
  journal={IEEE Access},
  year={2025},
  publisher={IEEE}
}
\bibliographystyle{IEEEtran}

\appendix

\subsection{Missing experimental results}

\begin{table*}[t]
\centering
\caption{Utility Results under Different Privacy Budgets}
\label{tab:epsilon}
\renewcommand{\arraystretch}{1.1}
\resizebox{0.85\linewidth}{!}{% resizebox只包裹tabular
\setlength{\tabcolsep}{4pt}
\begin{tabular}{c|cc|cc|cc|cc|cc} % 加竖线分隔Health与NQ
\toprule
& \multicolumn{2}{c}{HealthCareMagic} & \multicolumn{2}{c}{NQ} & \multicolumn{2}{c}{TQA} & \multicolumn{2}{c}{WQ} & \multicolumn{2}{c}{CT} \\
\cmidrule(lr){2-3} \cmidrule(lr){4-5} \cmidrule(lr){6-7} \cmidrule(lr){8-9} \cmidrule(lr){10-11}
$\epsilon$ & BLEU & ROUGE & BLEU & ROUGE-L & BLEU & ROUGE-L & BLEU & ROUGE-L & BLEU & ROUGE-L \\
\midrule
0.1 & 0.3231 & 0.3387 & 0.0666 & 0.1429 & 0.0638 & 0.1383 & 0.0592 & 0.1381 & 0.0623 & 0.1369 \\
1   & 0.3266 & 0.3445 & 0.0672 & 0.1435 & 0.0673 & 0.1444 & 0.0629 & 0.1429 & 0.0638 & 0.1387 \\
5   & 0.3275 & 0.3445 & 0.0674 & 0.1437 & 0.0688 & 0.1476 & 0.0666 & 0.1443 & 0.0638 & 0.1392 \\
10  & 0.3337 & 0.3501 & 0.0676 & 0.1452 & 0.0676 & 0.1435 & 0.0685 & 0.1437 & 0.0650 & 0.1409 \\
\bottomrule
\end{tabular}
}
\end{table*}

Table~\ref{tab:epsilon} reports the utility results on all datasets under different privacy budgets. Overall, the proposed method exhibits a clear privacy--utility trade-off: as the privacy budget $\epsilon$ increases, the utility metrics generally improve across all datasets. This trend is expected, since a larger privacy budget introduces less perturbation during the privacy-preserving generation process, allowing the generated responses to better preserve the semantic information of the retrieved content.

More importantly, the performance degradation under stringent privacy budgets remains relatively small. Even with a very strong privacy guarantee ($\epsilon=0.1$), the BLEU and ROUGE scores are only slightly lower than those obtained with much larger privacy budgets. For example, on the HealthCareMagic dataset, the ROUGE score increases from 0.3387 at $\epsilon=0.1$ to only 0.3501 at $\epsilon=10$, while similar trends are observed on the remaining datasets. This demonstrates that our method effectively preserves generation quality while providing rigorous differential privacy guarantees, making it suitable for privacy-sensitive RAG applications.

The relatively stable performance across a wide range of privacy budgets indicates that the proposed prompt-aware mechanism can accurately identify privacy-critical content, enabling the privacy budget to be utilized efficiently rather than perturbing all retrieved information uniformly.

\subsection{Differential Privacy Guarantee of PA-HDP}

We prove that PA-HDP satisfies differential privacy by decomposing the algorithm into three components: (i) Laplace perturbation for risk scoring, (ii) post-processing-based hierarchical filtering, and (iii) exponential mechanism for private candidate selection. The overall privacy guarantee follows from sequential composition and post-processing invariance.

\paragraph{Step 1: Laplace Mechanism for Risk Score.}
For each segment $s_{i,j}$, the risk score is computed as $S_{total}(q,s_{i,j})$. Since the risk score is normalized to the interval $[0,1]$, changing one neighboring segment under the semantic metric alters the score by at most one. Therefore, the semantic sensitivity of the risk score is bounded by $\Delta S = 1$.

To privately release the risk score, PA-HDP allocates half of the total privacy budget, i.e., $\epsilon/2$, and adds Laplace noise:
\[
S_{total}^{*}(q,s_{i,j})
=
S_{total}(q,s_{i,j})
+
\mathrm{Lap}\!\left(\frac{2}{\epsilon}\right).
\]

\begin{lemma}
The above mechanism satisfies $(\epsilon/2,0)$-Semantic Metric Differential Privacy.
\end{lemma}

\begin{proof}
Since the semantic sensitivity of the released risk score is $\Delta S=1$, adding Laplace noise with scale
\[
\frac{\Delta S}{\epsilon/2}=\frac{2}{\epsilon}
\]
guarantees $(\epsilon/2,0)$-Semantic Metric Differential Privacy according to the Laplace mechanism under the semantic metric.
\end{proof}

\paragraph{Step 2: Hierarchical Filtering (Post-processing).}
The algorithm applies thresholding:
\[
S^{\ast}_{total}(q, s_{i,j}) \geq \tau_1.
\]

Since this operation is a deterministic function of the DP output $S^{\ast}_{total}$, by the post-processing theorem of differential privacy, it does not incur additional privacy loss.

\paragraph{Step 3: Exponential Mechanism for Candidate Selection.} 

\noindent

\textbf{Lemma}~\ref{lem:exp}[Semantic Exponential Mechanism] 
Let $\mathcal C$ be a finite candidate text set, and let
\[
u(s_{i,j}, s_{i,j}^v)
=
- \alpha \cdot d_{\mathrm{sem}}(s_{i,j}, s_{i,j}^v)
- (1-\alpha) \cdot r(s_{i,j}, s_{i,j}^v),
\]
where $\alpha \ge 0$.
The Semantic Exponential Mechanism selects an output
$s_{i,j}^{\ast}\in\mathcal C$ according to the probability distribution
\[
\Pr[M(s_{i,j})=s_{i,j}^{\ast}]
=
\frac{
\exp\left(
\frac{\epsilon u(s_{i,j},s_{i,j}^{\ast})}
{2}%\Delta_u
\right)
}{
\sum\limits_{s_{i,j}^v\in\mathcal C}
\exp\left(
\frac{\epsilon u(s_{i,j},s_{i,j}^v)}
{2}
\right)
}.
\]
Then $M$ satisfies $(\epsilon,0)$-Semantic Metric Differential Privacy.

\begin{proof}

To prove that the mechanism satisfies Semantic Metric Differential Privacy, it suffices to show that the utility function is Lipschitz continuous with respect to the semantic metric $d_{\phi}$, namely,
\[
|u(s_{i,j},s_{i,j}^v)-u(s'_{i,j},s_{i,j}^v)|
\le
\Delta_u \, d_{\phi}(s_{i,j},s'_{i,j}).
\]

For the proposed utility function
\[
u(s_{i,j}, s_{i,j}^v)
=
- \alpha \cdot d_{\mathrm{sem}}(s_{i,j}, s_{i,j}^v)
- (1-\alpha) \cdot r(s_{i,j}, s_{i,j}^v),
\]
we first analyze the semantic preservation term:
\begin{align}
\Delta_{u(\mathrm{sem})}
&=
\left|
d_{\mathrm{sem}}(s_{i,j},s_{i,j}^v)
-
d_{\mathrm{sem}}(s'_{i,j},s_{i,j}^v)
\right|
\notag
\\
&=
\left|
\|\phi(s_{i,j})-\phi(s_{i,j}^v)\|_2
-
\|\phi(s'_{i,j})-\phi(s_{i,j}^v)\|_2
\right|
\notag
\\
&\le
\|\phi(s_{i,j})-\phi(s'_{i,j})\|_2
\notag
\\
&=
d_{\phi}(s_{i,j},s'_{i,j}).
\notag
\end{align}

Thus, the semantic preservation term is $1$-Lipschitz continuous with respect to $d_{\phi}$. Since this term is weighted by $\alpha$, its contribution to the overall sensitivity is bounded by $\alpha$.

Next, we analyze the privacy risk term:
\[
r(s_{i,j}, s_{i,j}^v)
=
\max_{e\in E(s_{i,j}^v)}
\cos(\phi(s_{i,j}),\phi(e)),
\]
where $E(s_{i,j}^v)$ denotes the set of sensitive entities extracted from the candidate text $s_{i,j}^v$. Assuming all embedding vectors are normalized, the cosine similarity is equivalent to the inner product. For two adjacent texts $s_{i,j}\sim s'_{i,j}$, we have
\begin{align}
\Delta_{u(\mathrm{risk})}
&=
|r(s_{i,j}, s_{i,j}^v)-r(s'_{i,j}, s_{i,j}^v)|
\notag
\\
&=
\Big|
\max_{e\in E(s_{i,j}^v)}
\cos(\phi(s_{i,j}),\phi(e))
-
\max_{e\in E(s_{i,j}^v)}
\cos(\phi(s'_{i,j}),\phi(e))
\Big|
\notag
\\
&\le
\max_{e\in E(s_{i,j}^v)}
|
(\phi(s_{i,j})-\phi(s'_{i,j}))^\top \phi(e)
|
\notag
\\
&\le
\|\phi(s_{i,j})-\phi(s'_{i,j})\|_2
\cdot
\|\phi(e)\|_2
\notag
\\
&\le
d_{\phi}(s_{i,j},s'_{i,j}).
\notag
\end{align}

Therefore, the privacy risk term is also $1$-Lipschitz continuous with respect to the semantic metric $d_{\phi}$. Since this term is weighted by $(1-\alpha)$, its contribution to the overall sensitivity is bounded by $(1-\alpha)$.

Combining both terms yields
\[
|u(s_{i,j}, s_{i,j}^v)-u(s'_{i,j}, s_{i,j}^v)|
\le
(\alpha+(1-\alpha))\, d_{\phi}(s_{i,j},s'_{i,j}).
\]

Hence, the global sensitivity of the utility function under the semantic metric is bounded by
\[
\Delta_u
\le
\alpha+(1-\alpha) = 1.
\]

Substituting the sensitivity bound into the exponential mechanism gives

\[
\Pr[M(s_{i,j}) = s_{i,j}^{\ast}]
\propto
\exp\left(
-\frac{\epsilon}{2}
\begin{aligned}[t]
\Big(
&\alpha \cdot d_{\mathrm{sem}}(s_{i,j}, s_{i,j}^{\ast}) \\
&+ (1-\alpha) \cdot r(s_{i,j}, s_{i,j}^{\ast})
\Big)
\end{aligned}
\right)
\]

which matches the formulation stated in the lemma. Therefore, the mechanism satisfies $(\epsilon,0)$-Semantic Metric Differential Privacy.

\end{proof}

\paragraph{Step 4: Composition Over Segments.}

Since the privacy budget $\epsilon/2$ is evenly allocated to the high-risk and medium-risk segments, by the sequential composition theorem, these two mechanisms together satisfy $(\epsilon/2,0)$-Semantic Metric Differential Privacy.

\paragraph{Theorem (Privacy Guarantee of PA-HDP).}
PA-HDP satisfies $(\epsilon,0)$-Semantic Metric Differential Privacy.

\begin{proof}
    The mechanism consists of:
(i) Semantic Laplace mechanism satisfying $(\epsilon,0)$-Semantic Metric Differential Privacy,
(ii) post-processing operations preserving DP,
(iii) Semantic exponential mechanism satisfying $(\epsilon,0)$-Semantic Metric Differential Privacy.

By sequential composition and post-processing invariance, the overall mechanism satisfies $(\epsilon,0)$-Semantic Metric Differential Privacy.
\end{proof}

\subsection{Details of Field Sensitivity Weights}

To instantiate the field sensitivity defined in Section X, we assign a predefined leakage risk weight $w_k \in [0,1]$ to each sensitive entity category. The weights are manually designed according to the potential privacy harm caused by disclosing each entity type, following the principle that entities enabling direct identity disclosure receive substantially higher scores than contextual or auxiliary information.

Specifically, we divide sensitive entities into two broad categories.

\paragraph{High-risk entities ($w_k \in [0.5,1.0]$)}
These entities can directly identify an individual or expose highly sensitive credentials (Table~\ref{tab:entities_high} ).

\begin{table}[h]
\centering
\small
\caption{High-risk entities}
\label{tab:entities_high}
\begin{tabular}{c|c|p{4.5cm}}
\toprule
Entity Type & Weight & Description \\
\midrule
ID Number & 0.95 & Unique personal identifier with severe privacy leakage risk \\
Phone Number & 0.90 & Direct contact identifier \\
Person Name & 0.90 & Primary identity attribute \\
Email Address & 0.75 & Online identity linkage \\
Password & 0.80 & Authentication credential \\
MED-HIGH & 0.85 & Diagnosis of high-risk diseases \\
\bottomrule
\end{tabular}
\end{table}

\paragraph{Low-risk entities ($w_k \in [0.1,0.5]$)}
These entities mainly provide contextual or auxiliary information (Table~\ref{tab:entities_low} ).

\begin{table}[h]
\centering
\small
\caption{Low-risk entities}
\label{tab:entities_low}
\begin{tabular}{c|c|p{4.5cm}}
\toprule
Entity Type & Weight & Description \\
\midrule
Location & 0.30 & Geographical contextual information \\
Organization & 0.40 & Affiliation or institutional context \\
Vital Signs & 0.20 & Basic physiological measurements \\
Amount & 0.50 & Financial context without identity disclosure \\
Date & 0.20 & Temporal information \\
MED-MED & 0.10 & Common disorder \\
\bottomrule
\end{tabular}
\end{table}

We emphasize that the specific numerical values are empirically defined and can be adapted to different application scenarios without modifying the overall framework.

\subsection{Implementation Details of Baseline Approaches}

We adopt the baseline methods and implementation protocols from the work of Zeng et al. \cite{8zeng2025mitigating}. The detailed implementation settings are described as follows.

\subsubsection{Paraphrase}
This approach employs large language models to extract relevant and essential components from the retrieved passages. Insignificant portions can be discarded, while selected sentences are rephrased to improve clarity or relevance. The instruction template we use for paraphrasing is provided in Table \ref{tab:prompt-paraphrase}.

\begin{table}[t]
\centering
\caption{Prompt of paraphrase}
\label{tab:prompt-paraphrase}
\renewcommand{\arraystretch}{1.0}
\setlength{\tabcolsep}{3pt}
\begin{tabular}{l}
\toprule
\textbf{Prompt} \\
\midrule
Given the following context, extract the useful or important part of the Context. \\
Remember, \textit{*DO NOT*} edit the extracted parts of the context. \\
\vspace{4pt}
$>$ Context: \\
$>>>$ \\
\texttt{\{input\_context\}} \\
$>>>$ \\
\vspace{4pt}
Extracted relevant parts: \\
\bottomrule
\end{tabular}
\end{table}

\subsubsection{ZeroGen}
The ZeroGen strategy is designed to produce a new set of question–answer pairs derived from the original passage. In practice, we first apply the spaCy toolkit to recognize named entity mentions in the source text. We then feed the LLM with the prompt: “The context is: {original context}. {extracted entities} is the answer to the following question: ” to solicit corresponding questions. The newly constructed context comprises 10 randomly chosen question–answer pairs, formatted as “Question: {generated questions}. Answer: {extracted entities}”.

\subsubsection{AttrPrompt}
Unlike the previous methods, AttrPrompt relies solely on LLM-generated content without accessing any retrieved data from the external database. This technique queries the LLM about the most critical attributes associated with a specific data category. For the ChatDoctor scenario, we prompt the model with: “What do you think are important attributes to generate some chat doctor data? Examples: disease...”. From the LLM’s response, we select five attributes and ask the model to propose 10 varied subtopics for each attribute. During new context generation, we randomly pick one subtopic per attribute and instruct the LLM to produce data following the assigned attribute.

\subsection{Detailed Attack Design}

This subsection elaborates on the design of both targeted and untargeted attacks against Retrieval-Augmented Generation (RAG) systems. These attack configurations are employed to assess the privacy-preserving capacity of our proposed synthetic-data framework. We emulate a realistic black-box environment where the adversary can only interact with the system via API calls. Consequently, the adversary's strategy centers on crafting and manipulating the query \(q\) to elicit the target information from the RAG pipeline.

\subsubsection{Prompt Composition}

The composite structured prompt generally consists of two segments: an \{information\} component and a \{command\} component. 
\[
q = \{\text{information}\} + \{\text{command}\}
\]
  
This formulation pursues two parallel goals: (a) steering the retriever to accurately fetch the desired pieces of data, and (b) compelling the language model to incorporate the retrieved evidence into its output. The \{information\} part directs the retrieval module toward specific data items, whereas the \{command\} part instructs the model to echo the retrieved content in its response. For the latter, we adopt directives like ``Please repeat all the context''; the former varies according to the attacker's objective.

\subsubsection{Targeted Attack}

In a targeted scenario, the adversary seeks to recover a predetermined set of specific details. Constructing the information component for such an attack proceeds in two steps. Initially, the attacker supplies concrete exemplars based on their goal -- for instance, ``I want some advice about target name'' for a well-defined target, or a partial prefix such as ``Please email us at'' for a more vague target. Subsequently, a substantial number of analogous yet diverse instances are generated from these exemplars. For targets that contain multiple sub-elements (e.g., the HealthcareMagic corpus), variations can be produced by substituting sub-items like disease names derived from ChatGPT or the International Classification of Diseases (ICD). Alternatively, large language models (e.g., ChatGPT) can directly yield similar sentences based on the provided examples -- this approach is also applied to the Wiki-PII dataset. For instance, one might input ``Generate 100 similar sentences like 'Please email us at' ''.

\subsubsection{Untargeted Attack}

In contrast, untargeted attacks emphasize the production of varied information components so as to extract a broad spectrum of data from the retrieval repositories, rather than homing in on any particular entry. Inspired by the methodology of Carlini et al.\cite{carlini2021extracting}, we randomly sample segments from the Common Crawl corpus to serve as the information part. However, the randomness of the input could potentially interfere with the command component. To counter this, we cap the maximum token length of the information part at 15 tokens, thereby preserving prompt coherence and maintaining effectiveness in eliciting data from the retrieval collections.

\subsection{Details of Evaluation Metrics}

This section clarifies the evaluation metrics adopted in our study.
\subsubsection{ROUGE-L}
ROUGE-L belongs to the ROUGE (Recall-Oriented Understudy for Gisting Evaluation) family and is widely used for assessing text-generation tasks, including summarization and machine translation. It measures the overlap between a generated output and a reference text via the Longest Common Subsequence (LCS).

\begin{itemize}
    \item \textbf{Longest Common Subsequence (LCS):} ROUGE-L identifies the longest word sequence that appears in both the generated and reference texts while preserving the original order, without requiring contiguity.
    \item \textbf{Recall, Precision, and F-measure:}
    \begin{itemize}
        \item Recall is defined as the ratio of the LCS length to the reference length (\(n\)): \(\text{Recall} = \text{LCS}(X,Y) / n\). It reflects the fraction of the reference sequence captured by the generator.
        \item Precision is the ratio of the LCS length to the generated length (\(m\)): \(\text{Precision} = \text{LCS}(X,Y) / m\). It indicates how much of the generated content aligns with the reference.
        \item The F-measure, which balances precision and recall via their harmonic mean, is computed as:
        \[
        F_{\text{lcs}} = \frac{(1+\beta^2) \cdot R_{\text{lcs}} \cdot P_{\text{lcs}}}{R_{\text{lcs}} + \beta^2 \cdot P_{\text{lcs}}}
        \]
        where \(\beta\) controls the trade-off between precision and recall (typically set to \(1.0\)). In our experimental results, we report the F-measure as the ROUGE-L score.
    \end{itemize}
\end{itemize}

Let \(C\) denote the candidate translation, and \(R\) represent the set of reference translations.

\subsubsection{BLEU-1}

BLEU-1 evaluates translation quality based on unigram precision.

\begin{itemize}
    \item \textbf{Unigram precision:}
    \[
    P_1 = \frac{\sum_{w} \min\bigl(\text{Count}_C(w), \max\text{Count}_R(w)\bigr)}{\sum_{w} \text{Count}_C(w)}
    \]
    where \(\text{Count}_C(w)\) is the frequency of word \(w\) in the candidate, and \(\max\text{Count}_R(w)\) is the maximum frequency of \(w\) across any single reference translation.
    \item \textbf{Brevity penalty (BP):}
    \[
    \text{BP} = \min\left(1, \exp\left(1 - \frac{r}{c}\right)\right)
    \]
    where \(c\) is the candidate length and \(r\) is the reference length that is closest to \(c\).
    \item \textbf{Final BLEU-1 score:}
    \[
    \text{BLEU-1} = \text{BP} \times P_1
    \]
    The score ranges from 0 to 1, with 1 indicating perfect unigram matching against the reference.
\end{itemize}

\subsubsection{Additional Metrics}

Beyond the above, we incorporate several new metrics to further corroborate our approach:
\begin{itemize}
    \item \textbf{Exact Match (EM):} This metric checks whether the ground-truth answer appears verbatim within the LLM's generated response.
    \item \textbf{LLM-based Correctness Judgment:} We employ Ragas, a widely adopted automatic evaluation pipeline for RAG (with over 5.9k stars on GitHub), to assess the factual correctness of the generated answers.
\end{itemize}

\subsection{Construction of Wiki-PII Dataset}

To assess the privacy-preserving performance of our proposed approach against targeted extraction attempts, we adopt the dataset construction methodology introduced in \cite{8zeng2025mitigating} to create the Wiki-PII dataset.
This corpus is specifically designed to contain a substantial volume of personally identifiable information (PII), thereby providing a rigorous testbed for evaluating privacy-protection strategies. The construction procedure comprises three main phases.

\begin{enumerate}
    \item \textbf{PII extraction:} We first harvested real PII instances from the Enron Mail corpus. Specifically, we employed the \texttt{urlextract} library to extract web URLs and applied regular expressions to capture phone numbers and personal email addresses.
    \item \textbf{Text chunking:} In the second phase, we partitioned the Wikipedia text collection using the recursive character text splitter provided by LangChain, with a chunk size of 1500 characters.
    \item \textbf{PII insertion:} Finally, for each resulting text chunk, we randomly appended the extracted PII items to the end of every sentence within that chunk.
\end{enumerate}

\subsection{Dataset Preparation Details}

In our experiments, we adopt the configuration established in prior RAG privacy research \cite{ex-huang2023privacy} to simulate a realistic scenario where a large-scale public repository inadvertently contains certain private information — for instance, personal names or phone numbers that appear within broadly accessible web content. In practice, virtually any sentence in a text corpus may carry privacy-sensitive entities. Consequently, embedding PII into the dataset serves as a plausible and effective strategy for modeling such real-world privacy risks.

\subsection{Examples of Protected Contexts}

Table~\ref{tab:rag_compare} presents representative examples of the protected query--response pairs in the retrieval corpus. Compared with the original records, the protected versions effectively obfuscate privacy-sensitive information while preserving the semantic content required for downstream retrieval. Specifically, sensitive entities, including medication names, disease descriptions, temporal information, numerical values, and demographic attributes, are replaced with plausible alternatives or generalized expressions, making it substantially more difficult to recover the original records or infer personal information.

Despite these perturbations, the protected query--response pairs largely retain their original semantics. In the first example, although several medications and personal details are replaced, the protected record still clearly describes a fertility-related consultation and its corresponding medical advice. Similarly, in the second example, detailed medical history and treatment information are concealed, while the protected response continues to capture the clinical recommendation regarding pregnancy risk. These examples demonstrate that our protection mechanism effectively removes privacy-sensitive content while preserving the essential medical knowledge required for retrieval.

We also observe that stronger perturbations may occasionally introduce minor semantic deviations or less natural expressions, as illustrated in the third example. Nevertheless, the overall medical intent and question--answer correspondence remain largely intact, ensuring that the protected corpus continues to provide useful evidence for retrieval-augmented generation. These qualitative examples further demonstrate that our method achieves a favorable balance between privacy protection and retrieval utility.

\begin{table*}[htbp]
  \centering
  \caption{Examples of Protected Contexts}
  \label{tab:rag_compare}
  % 两列，每列固定宽度6.8cm，长文本自动换行
  \begin{tabular}{p{6.8cm} p{6.8cm}}
    \hline
    Original context & Protected context \\
    \hline
    input: yes, can you tell me if lisinopril hctz , omeprazole , and simvastain effect my husbands sperm or fertility ? we are trying to get pregnant and have been for 9yrs and we have been to dr after dr test after test ... and no answers. so im doing some findings and wanting to know about his meds. thanks
    
    output: Hi, Thanks for writing to Chat Doctor. Let me explain you one by one. Surprisingly there are fewer studies or researches on this topic. Omeprazole is not know to cause any adverse effect on sperm count or fertility. Lisinopril is not having direct effect on sperm count, but it effects fertility by following side effects---Decreased interest in sexual intercourse-inability to have or keep an erection-lack or loss of strength-loss in sexual ability, desire, Chat Doctor. But these are rare side effects. Simvastin rarely cause erectile dysfunction. It is also shown in some studies that it reduce testosterone levels which is required for sperm production. This may be the cause of reduced sperm count. You can discuss with your doctor to switch to alternative.
    &
    input: yes, can you tell me if Eurax hctz , breakage of acid bottle , and smoke effect my husbands sperm or fertility. we are trying to get pregnant and have been for 6.8yrs and we have been to dr after dr test after test. and no answers. so im doing some findings and wanting to know about his meds. thanks.
    
    output: hi, thanks for writing to chat doctor. let me explain you one by one. surprisingly there are fewer studies or researches on this topic. palmetto is not know to cause any adverse effect on sperm count or fertility. lisinopril is not having direct effect on sperm count, but it effects fertility by following side effects---decreased interest in sexual intercourse-inability to have or keep an erection-lack or loss of strength-loss in sexual ability, desire, chat doctor. but these are rare side effects. simvastin rarely cause erectile dysfunction. it is also shown in some studies that it reduce testosterone levels which is required for sperm production. this may be the cause of reduced sperm count. you can discuss with your doctor to switch to alternative.  \\ \hline
    input: My wife has pulmonary fibrosis 2yrs now doing fair,cough present but functional.this past feb. she had her ascend colon removed secondary to colon cancer stage 1c(2).no chemo was needed....a month ago she turns up pregnant. My question is how dangerous is it for her to have this child.she is 42,not in great shape from 2 yrs of of meds ,prednisone(ended in April).
    
    output: Hi and welcome to Chat Doctor, thank you for your query. Considering all this, her pregnancy may be risky, and you should consult gynecologist before you decide to have baby. Her condition is serious, and it may have bad effect on possible pregnancy and baby. So I wouldst suggest going on without more detailed consultation with gynecologist. Wish you good health. Regards.
    &
    input: My wife has pulmonary WAS EXPERIENCING Sunday now doing fair,irritable present but functional. this past feb. she had her ascend colon removed secondary to cancer 1c(2). no chemo was needed. a week she turns up pregnant. my question is how dangerous is it for her to have this child. she is 42,not in great shape from 2 yrs of of meds ,prednisone(ended in april). 
    
    output: Hi and welcome to chat doctor, thank you for your query. considering all this, her pregnancy may be risky, and you should consult gynecologist before you decide to have baby. her condition is serious, and it may have neurotoxic on possible pregnancy and baby. so i wouldst suggest going on without more detailed consultation with gynecologist. wish you good health. regards.
    \\ \hline
    input: Hi my daughter of 4.5 years old is suffering from cough and fever since yesterday morning. Can I give her zifi 100 mg ready syrup . Last year doctor had advised her same antibiotic with crocin and for fever and cough. What should be the dosage for her.
    
    output: Hello, I am sorry your daughter is ill. With the symptoms you listed-cough and fever- it is likely that she has a viral respiratory illness.  There are more viruses than bacteria that cause respiratory infections. How high is her fever?  Is her appetite normal? Does she have her usual level of activity?  The severity of a viral respiratory illness is usually less that for a bacterial respiratory illness. WiFi is a brand name for the antibiotic called Definite.  While this is a good antibiotic to treat various bacterial infections, it does not treat viral infections. In fact, no antibiotic does. Plus, WiFi, like all Chat Doctor.  Some of its side effects are:- headache-loss of appetite-bloody or watery stool-dizziness-itchy rash-swelling of face, lips, tongue So for the symptoms you provided, this medicine most probably will not be effective, and it could cause side effects. Based on the information given, I feel it is not a good idea to give her this medicine.  Hope I have answered your query. Let me know if I can assist you further.
    &
    input: hi my daughter of 4. 22nd April is suffering from blisters of 0 and 3 days around 3 minutes. Can i give her zifi 94 mg ready syrup. 9 weeks doctor had advised her same antibiotic with breakage of acid bottle for bad nappy rash and cough. what should be the dosage for her. 
    
    output: hello, i am sorry your daughter is ill. with a brain tumor you listed-cough and fever- it is likely that she has a viral allergic reaction. there are more viruses than bacteria that cause ascitis. how high is her cough. is her appetite normal. does she have her usual level of activity. the severity of a viral respiratory illness is usually less that for a bacterial respiratory illness. wifi is a brand name for the antibiotic called definite. while this is a good antibiotic to treat various bacterial infections, it does not treat bad back pain. in fact, no antibiotic does. plus, wifi, like all chat doctor. some of its side effects stomach pains of abd painf face, lips, tongue so for this pain medication you provided, this medicine most probably will not be effective, and it could cause side effects. based on the information given, i feel it is not a good idea to give her this medicine. hope i have answered your query. let me know if i can assist you further. 
    \\ \hline
  \end{tabular}
\end{table*}

\subsection{Discussions when Adapting PA-HDP to Domain-Specific Applications}

Although PA-HDP is presented in the context of medical and QA RAG, the overall framework is domain-independent and can be readily adapted to other application scenarios. The primary modification lies in the sensitive entity extraction module, as different domains involve different categories of privacy-sensitive information.

Specifically, our framework relies on the extracted sensitive entities to estimate field sensitivity and subsequently allocate privacy budgets. Therefore, when applying PA-HDP to a new domain, the Named Entity Recognition (NER) model and rule-based extraction patterns should be customized according to the domain-specific privacy taxonomy. For example, in the financial domain, sensitive entities may include bank account numbers, credit card numbers, transaction records, and customer identifiers, while in the legal domain, case numbers, client identities, and confidential agreements may require protection. Existing domain-specific NER models can be directly integrated into our framework, or custom entity recognizers can be trained when annotated data are available.

In addition, the predefined sensitivity weights assigned to different entity types should be adjusted to reflect their privacy risks in the target domain. Since the subsequent privacy budget allocation only depends on the extracted entity categories and their corresponding sensitivity scores, no modification to the remaining components of PA-HDP is required. This modular design enables PA-HDP to be easily extended to diverse privacy-sensitive RAG applications while preserving its query-aware privacy protection mechanism.

\end{document}